\documentclass[amsmath,amssymb,aps,prl,twocolumn,nofootinbib,superscriptaddress,longbibliography]{revtex4-2}

\usepackage[english]{babel}
\usepackage[utf8]{inputenc}
\usepackage[normalem]{ulem}
\usepackage{amsmath,amssymb,amsfonts,dsfont,physics,graphicx,xcolor,mathtools,bm,hyperref}
\usepackage{amsthm}
\usepackage{algorithm,algpseudocode}
\usepackage{booktabs}
\usepackage{tikz}
\usetikzlibrary{positioning,arrows.meta,shapes.geometric,decorations.pathmorphing}

\newtheorem{proposition}{Proposition}
\newtheorem{remark}{Remark}
\hypersetup{
	colorlinks=true,
	linkcolor={red!40!black},
	citecolor={blue!60!black},
	urlcolor={blue!50!black}
	}
\graphicspath{{fig/}}
\definecolor{myred}{rgb}{0.85,0,0}
\definecolor{mygray}{rgb}{0.87, 0.87, 0.87}

\begin{document}

\title{Quantum Randomized Subspace Iteration}

\author{Stefano Scali}
\email{stefano.scali@fujitsu.com}
\affiliation{Fujitsu Research of Europe Ltd., Slough SL1 2BE, UK}
\affiliation{University of Exeter, Department of Physics and Astronomy, Stocker Road, Exeter EX4 4QL, UK}

\author{Brian Coyle}
\affiliation{Fujitsu Research of Europe Ltd., Slough SL1 2BE, UK}
\affiliation{School of Informatics, University of Edinburgh, Edinburgh, UK}

\author{Giuseppe Buonaiuto}
\affiliation{Fujitsu Research of Europe Ltd., Pozuelo de Alarcón, 28224 Madrid, Spain}

\author{Michal Krompiec}
\affiliation{Fujitsu Research of Europe Ltd., Slough SL1 2BE, UK}

\begin{abstract}
Resolving degenerate quantum eigenspaces—including topologically
ordered ground states and frustrated magnets—requires preparing
high-fidelity states that span every direction of the target manifold.
Existing variational and projective algorithms do not naturally cover a
multi-dimensional degenerate subspace without sequential
orthogonality constraints. We introduce the quantum randomized
subspace iteration (QRSI), a fully parallel construction that
conjugates the Hamiltonian by independent random unitaries across as
many branches as the degeneracy~$g$, then invokes any chosen
eigenstate-preparation primitive on each branch. The target subspace
is identified from the resulting ensemble via standard subspace
estimation, either classically through the coefficient matrix or on
hardware through Gram-matrix measurements. We prove that the construction
spans the full eigenspace almost surely
and preserves the spectral gap exactly on every branch. For practical
use, we show that these guarantees hold whenever the random rotations
satisfy an anti-concentration condition over the degenerate manifold,
substantially weaker than full Haar randomness. We
demonstrate QRSI on the toric code, recovering all four topological
ground states, and on random Hamiltonians with planted degeneracies.
\end{abstract}

\maketitle



Preparing states spanning a degenerate eigenspace of a quantum
Hamiltonian is central to quantum science, from detecting topological
phases in frustrated
magnets~\cite{balents2010_spin,savary2016quantum,Zhang2012,Cincio2013}
to resolving near-degeneracies in correlated
systems~\cite{helgaker2000molecular} and computing Betti numbers in
quantum topological data analysis
(QTDA)~\cite{LGZ2016,gyurik2022towards,ScaliUmeanoKyriienko2024}.
Even classical methods need problem-specific machinery for this:
accessing the full multi-sector ground space of a topologically
ordered Hamiltonian has required adiabatic flux insertion on
infinite cylinders in DMRG~\cite{HeShengChen2014,Cincio2013,
StoudenmireWhite2012}, tied to particular lattice geometries.
A primitive-agnostic, hardware-agnostic counterpart has been missing.
The centrality of eigenspace resolution carries over to any observable
depending on multiple eigenvectors: kernel dimension, Gram matrices, and structure factors.
The object of interest is a target degenerate eigenspace
$\mathcal{G} = \operatorname{eigenspace}(H,E_\star) =
\operatorname{span}\{\ket{v_1},\dots,\ket{v_g}\}$, e.g. with $E_\star = E_0$
as the canonical (ground) instance and dimension $g$ possibly unknown.
Given an ensemble $\{\ket{\psi_i}\}_{i=1}^M$, two requirements must be
met simultaneously:
(i)~\emph{high overlap},
$\|\Pi_{\mathcal{G}}\ket{\psi_i}\|^2 = O(1)$ for each~$i$
(i.e., bounded below by a constant independent of~$N$); and
(ii)~\emph{high diversity}, meaning the projected vectors
$\{\Pi_{\mathcal{G}}\ket{\psi_i}\}$ span~$\mathcal{G}$.
No standard primitive achieves selectability without sequential
inter-branch coupling (orthogonality penalties, deflation, or
Lagrange multipliers between successive states).
Call a probe set $\mathcal{R}\subset\mathcal{H}$
\emph{reachable} for $\mathcal{G}$ if
$\Pi_{\mathcal{G}}\mathcal{R}$ spans $\mathcal{G}$,
and a preparation map $\mathcal{P}:\mathcal{H}\to\mathcal{R}$
\emph{selectable} if some $g$ outputs span $\mathcal{G}$.
Unrandomized preparation typically achieves reachability but not
selectability: variational optimisers converge to one basin per
invocation, while projective methods with random seeds funnel to a
single dominant direction whenever the eigenvectors are structured
(stabilizer codes, frustrated magnets; Supplemental Material,
Remark~\ref{rem:eigvec_structure}). Random probes are trivially
selectable but only at the Haar baseline overlap $O(g/N)$.
Workarounds (SSVQE~\cite{nakanishi2019subspace},
VQD~\cite{higgott2019variational},
MC-VQE~\cite{parrish2019quantum},
QSE~\cite{mcclean2017hybrid}) target one eigenvalue at a time and
require a non-degenerate spectrum or hand-chosen operator basis.
We call this the \emph{diversity/overlap tradeoff}.

We introduce \emph{quantum randomized subspace iteration}
(QRSI), a framework that breaks this tradeoff via three
primitive-independent ingredients: \emph{state preparation}
concentrating weight on the target subspace, a \emph{random rotation}
injecting diversity over multiple \emph{branches}, and an \emph{overlap amplifier} sharpening the
$\mathcal{G}$-projection. Fig.~\ref{fig:numerics} demonstrates QRSI on the toric code
Hamiltonian~\cite{Kitaev2003}, recovering all four degenerate ground
states and resolving degeneracy at every excited level. In the below, we describe the mechanism and formal results for QRSI.
\begin{figure*}[htb!]
  \centering
  \includegraphics[width=\textwidth]{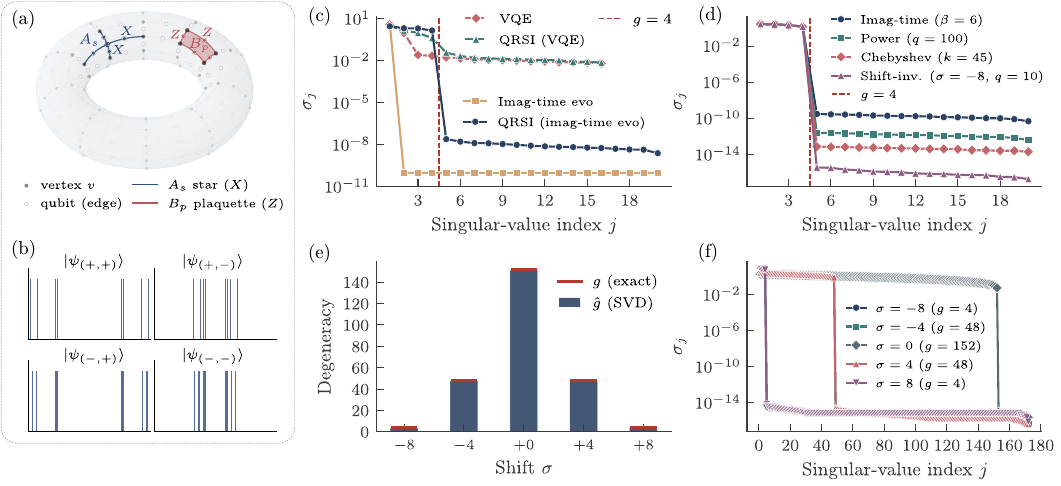}
  \caption{\textbf{QRSI on the perturbed $2{\times}2$ toric code
    ($N=256$, $g=4$, 8 qubits).}
    \textbf{(a)}~Stabilizer Hamiltonian (vertex and plaquette terms).
    \textbf{(b)}~Four topologically distinct ground states at
    $E_0\equiv\sigma=-8$.
    \textbf{(c)}~Singular-value spectrum with (navy) and without
    (slate) Haar rotation; the SVD gap at $j=g$ appears only with
    rotation.
    \textbf{(d)}~Four preparation primitives (imaginary time, power
    iteration, Chebyshev filtering, shift-and-invert) all recover the
    same degeneracy.
    \textbf{(e)}~Detected degeneracy $\hat{g}$ (bars) versus exact
    $g$ (red markers) at each target energy level via shift-and-invert
    centred on that level.
    \textbf{(f)}~Full SVD spectra at each level.}
  \label{fig:numerics}
\end{figure*}

\emph{Setup.} Consider $H \in \mathbb{C}^{N \times N}$ with $g$-fold
degenerate ground energy $E_0$ and spectral gap
$\gamma = E_{g+1} - E_0 > 0$, encoded on
$n = \lceil \log_2 N \rceil$ qubits
(when $N < 2^n$ the Hilbert space is embedded via a penalty term;
see Supplemental Material, Sec.~\ref{supp:padding}).
For each ensemble member $i = 1,\dots,M$ we sample a random unitary
$R_i \in U(N)$ and form the rotated Hamiltonian
$H_i = R_i^\dagger\,H\,R_i$.
A \emph{state-preparation stage} returns a trial state from a reachable
set $\mathcal{R} \subset \mathcal{H}$; an \emph{overlap amplifier}
returns coefficients with larger $\mathcal{G}$-weight.
Note that overlap amplifier and state-preparation maps can coincide.
The basis correction
$\bm{c}_i = R_i\,\tilde{\bm{c}}_i$
maps each amplified state back to the original frame.
Since $\mathcal{G}_i \coloneqq \operatorname{eigenspace}(H_i,E_0) = R_i^\dagger\,\mathcal{G}$,
each branch presents the
preparation stage with a randomly oriented copy of~$\mathcal{G}$.
At the branch level the construction can be written as
\begin{equation}\label{eq:qrsi_branch}
  \ket{\psi_i}
  = R_i\,\mathcal{P}^{(q)}_{\omega_i}\!\left(R_i^\dagger H R_i;\,\mathcal{R}\right),
\end{equation}
where $\mathcal{P}^{(q)}_{\omega_i}$ denotes the chosen preparation-plus-amplification
primitive, including any internal randomness $\omega_i$ (see Supplemental
Material, Fig.~\ref{fig:schematic}).
The construction admits two equivalent pictures: conjugating $H$ by
$R_i$ before preparation (Hamiltonian-rotation), or keeping $H$ fixed
and rotating the reachable set to $\mathcal{R}_i = R_i\mathcal{R}$
(state-rotation). They are unitarily equivalent branch by branch and
the choice is set by the constraints of the chosen primitive.

\begin{figure}[t]
    \centering
    \includegraphics[width=\linewidth]{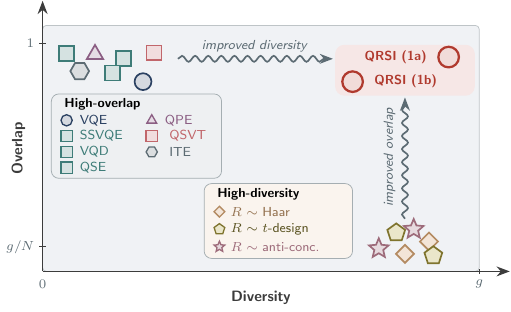}
    \caption{\textbf{Cartoon of the diversity / overlap landscape for
    quantum degenerate-eigenspace methods.}
    Schematic, not to scale. Variational and subspace-expansion methods
    (VQE family, QSE) achieve high overlap but low diversity: the
    optimiser's basin of attraction selects a single ground-state
    direction per invocation. Random probing achieves high diversity but
    only the Haar baseline overlap $O(g/N)$. The upper-right corner
    (high overlap and high diversity) is unoccupied by existing quantum
    methods. QRSI fills this corner by injecting Haar-random Hamiltonian
    conjugation into the preparation loop
    (Proposition~\ref{prop:diversity}), with a quantitative relaxation to
    anti-concentrated rotation ensembles
    (Proposition~\ref{prop:diversity_anti_concentration}). Classical randomised
    subspace iteration occupies the analogous corner for matrix
    algorithms; QRSI is its quantum analogue at the level of
    state-preparation primitives.}
    \label{fig:landscape}
\end{figure}

QRSI is the quantum analog of
classical randomized subspace iteration
(RSI)~\cite{halko2011finding,musco2015randomized,Saad2011eigenvalue},
which recovers invariant subspaces by sketching with a random probe
matrix $\Omega \in \mathbb{R}^{d \times \ell}$, amplifying via a
polynomial or rational spectral filter (power iteration $A^q\Omega$,
shift-and-invert $(A - \sigma I)^{-q}\Omega$, or Chebyshev
filtering~\cite{ZhouSaad2007}) and extracting the rank from the SVD
of the resulting sketch.
In QRSI the random rotations $R_i$ supply the probe directions;
preparation and amplification together provide the filter power that
boosts overlap from the random baseline to $O(1)$; and the SVD of the
coefficient matrix performs the rank-revealing factorization.
The analogy extends to convergence rates.
Classically, oversampling by
$\ell - g \ge 2$ columns and applying $q$ filter steps yields
approximation error decaying as
$(\gamma/\|A\|)^{2q+1}$ for power
iteration~\cite{halko2011finding}, with Chebyshev acceleration
improving the contraction to $\sim(\gamma/\|A\|)^{2q}$ at better
constants~\cite{ZhouSaad2007}.
In QRSI, the oversampling margin $M - g$ controls the probability of
spanning (Proposition~\ref{prop:diversity}), while the amplification
depth $q$ enters through the per-branch excited-state leakage
$\varepsilon_q \coloneqq \max_i\|\Pi_{\mathcal{G}^\perp}\ket{\psi_i}\|^2$,
the worst-case ground-space defect of the prepared branch states.
For imaginary-time evolution
$\varepsilon_q\sim e^{-2q\gamma}$; for QSVT eigenstate filtering
$\varepsilon_q\sim\delta^2$, where $\delta$ is the polynomial
approximation error.

The parameter~$q$ maps to primitives in two ways.
For \emph{iterative} methods (imaginary-time evolution, VQE,
adiabatic preparation), $q$ counts discrete steps, with
$C_{\mathrm{prep}} = O(q\gamma^{-1})$ for imaginary-time evolution.
For \emph{single-shot} methods (QPE~\cite{Kitaev1995, Abrams1999, PoulinWocjan2009},
QSVT eigenstate filtering~\cite{LowChuang2017, Gilyen2019, Lin2020, Dong2022}),
the entire subroutine acts as a single resolvent power at
effectively infinite~$q$: one application projects onto the
target eigenspace to accuracy~$\delta$, so preparation and
amplification collapse into a single call.

\emph{Primitives, pictures, and spectral reach.} The
Hamiltonian-rotation and state-rotation pictures are not equally
convenient for every primitive.
Spectral-filter methods (polynomial/Chebyshev filtering and Krylov
iteration) are most natural in the \emph{state-rotation} picture,
where $H$ stays fixed and $R_i$ acts only on the probe seed.
Adiabatic preparation is naturally \emph{Hamiltonian-rotation}, since
the endpoint $H_i = R_i^\dagger H R_i$ depends on~$R_i$.
Variational and projective methods work in either picture, but
state-rotation is often preferable because it avoids re-expanding
$H_i$ in Pauli strings.
Examples of primitive mapping are given in
Supplemental Material, Sec.~\ref{supp:primitives}.
Also, QRSI is not restricted to ground-state targeting.
The shift-and-invert transformation replaces the preparation filter by
$(H - \sigma I)^{-q}$, amplifying whichever eigenvalue is closest to
the shift~$\sigma$.
By sweeping $\sigma$ across the spectrum one turns QRSI into a
\emph{spectral microscope} that reads off each level's degeneracy from
the SVD gap.
Alternatively, the folded-spectrum method
$\tilde{H} = (H-\sigma I)^2$ reduces the problem to a ground-state
computation of~$\tilde{H}$, to which any of the primitives above
apply directly.
Chebyshev band-pass filters~\cite{GeTuraCirac2019} and symmetry-sector projections provide
additional options (Supplemental Material, Sec.~\ref{supp:primitives}).

\emph{Geometric interpretation.} On $\mathbb{C}P^{N-1}$ with the
Fubini--Study metric~\cite{Provost_1980,Gibbons_1992,Brody_2001,Cocchiarella_2020},
the preparation family traces a measurable subset $\mathcal{M}$ and
the optimizer's basin of attraction selects a single foot-point on
$\mathbb{C}P(\mathcal{G})$. Unitary conjugation
$[\psi]\mapsto[R\psi]$ is a Fubini--Study isometry, so QRSI
replaces $\mathcal{M}$ by $R_i\mathcal{M}$, preserving the
per-branch optimization landscape while moving the foot-point
across $\mathbb{C}P(\mathcal{G})$, upgrading local tangent-space
exploration to global isometric exploration. The full development
is in Supplemental Material, Sec.~\ref{supp:geometry}, with the
manifold-level identity for overlap preservation.

A naive serial scheme (prepare, then rotate the output) fails:
by left-invariance of the Haar measure, rotating a prepared state
erases its ground-space overlap, returning the effective overlap to
the random baseline $O(g/N)$ (see Supplemental Material, Sec.~\ref{supp:serial}).
QRSI circumvents this by placing the rotation \emph{inside} the
preparation--amplification loop, so that each branch re-optimises on
a freshly oriented instance and delivers $O(1)$ overlap ab initio.

QRSI inherits the per-call cost of whichever preparation primitive
is selected on each branch and does not amplify it beyond the $M$
parallel invocations already counted in the total
$M\cdot C_{\mathrm{prep}}$. Two implementation considerations,
sparsity destruction in the Hamiltonian-rotation picture and
primitive-specific costs in the state-rotation picture, are detailed
in Supplemental Material,
Sec.~\ref{supp:practical_challenges}.

\emph{The QRSI pipeline.} For each $i = 1,\dots,M$ independently,
sample~$R_i$, prepare on the rotated instance, amplify, and correct
back.
Every branch solves the same spectral problem with a different global
orientation of $\mathcal{G}$; the corrected ensemble spans $\mathcal{G}$
almost surely for $M \ge g$
(Proposition~\ref{prop:diversity}).

The central property underpinning QRSI is
the following.

\renewcommand{\theproposition}{1a}
\begin{proposition}[Diversity]\label{prop:diversity}
Let $R_1,\dots,R_M$ be independent Haar-random unitaries on $U(N)$,
and suppose the chosen preparation--amplification rule returns
$\ket{\tilde\phi_i}$ with nonzero ground-space component
$\Pi_{\mathcal{G}_i}\ket{\tilde\phi_i} \neq 0$ almost surely (a.s.) on every
branch.
Define the foot-point matrix $F\in\mathbb{C}^{g\times M}$ by
$F_{k,i} = \braket{v_k}{R_i|\tilde\phi_i}$.
Then for $M \ge g$, $\operatorname{rank}(F) = g$ a.s.;
equivalently, the $\mathcal{G}$-projected ensemble
$\bigl\{\Pi_{\mathcal{G}}\,R_i\ket{\tilde\phi_i}\bigr\}_{i=1}^M$
spans $\mathcal{G}$ a.s..
\end{proposition}

\begin{proof}[Proof sketch]
Let $\omega_i$ denote the internal randomness of the $i$-th
preparation--amplification call (e.g.\ random parameter initialization
in VQE, random measurement outcomes in QPE, or random Trotter seeds),
assumed independent of $R_i$.
By encoding all randomness in $\omega_i$, the preparation stage
is a fixed map for each~$\omega_i$:
$\mathcal{P}_{\omega_i}: H_i \mapsto
\ket{\tilde\phi_i}$.
Write $\ket{\tilde\phi_i} = \alpha_i\ket{v_i^*} + \beta_i\ket{\xi_i}$
with $\ket{v_i^*}\in\mathcal{G}_i$,
$\ket{\xi_i}\in\mathcal{G}_i^\perp$ normalised;
the hypothesis gives $\alpha_i\neq 0$ a.s.
Since $R_i\ket{\xi_i}\in\mathcal{G}^\perp$, the $i$-th column of $F$
satisfies $F_{k,i} = \alpha_i\,\braket{v_k}{R_i|v_i^*}$;
we call $f(R_i,\omega_i) \coloneqq R_i\ket{v_i^*} \in S^{2g-1}
\subset\mathcal{G}$ the \emph{foot-point}.
For any unitary~$S$ acting as $S_g \in U(g)$ on $\mathcal{G}$ and as
identity on $\mathcal{G}^\perp$, we have $[S,H] = 0$ and hence
$H_{SR_i} = H_i$; at fixed $\omega_i$ the preparation output is
unchanged, but $R_i\ket{v_i^*}\mapsto SR_i\ket{v_i^*}$.
Since $SR_i$ is Haar-distributed whenever $R_i$ is, the conditional
direction of $R_i\ket{v_i^*}$ given $\omega_i$ is
$U(g)$-invariant and therefore uniform on
$S^{2g-1} \subset \mathcal{G}$.
Because this holds for every realisation of $\omega_i$, the
unconditional distribution is also uniform on
$S^{2g-1}$ (tower property).
Since the pairs $(R_i,\omega_i)$ are mutually independent across
branches and $\alpha_i\neq 0$ a.s., for $M \ge g$ the columns of $F$
are independent draws from a continuous distribution on
$\mathbb{C}^g$; their linear dependence has measure zero, giving
$\operatorname{rank}(F) = g$ a.s..
The complete proof, including the $U(g)$-equivariance argument and the
finite-accuracy extension, is given in Supplemental Material,
Sec.~\ref{supp:proof:diversity}.
\end{proof}

\begin{remark}[Polynomial-structure cancellation]
\label{rem:cancellation}
The proof contains a structural cancellation worth isolating.
Because $\ket{v_i^*}\in\mathcal{G}_i = R_i^\dagger\mathcal{G}$, we
have $R_i\ket{v_i^*} = R_i R_i^\dagger\ket{w_i} = \ket{w_i}$ for
some $\ket{w_i}\in\mathcal{G}$, so the foot-point matrix entries
reduce to $F_{k,i} = \alpha_i\braket{v_k}{w_i}$. The explicit $R_i$
in the correction step exactly cancels the implicit $R_i^\dagger$
in the rotated ground-state basis: \emph{all polynomial dependence
on matrix elements of $R_i$ disappears}.
Consequently, standard moment-matching tools ($t$-designs) control
the \emph{input} overlaps $\braket{v_k}{R_i|\psi_0}$ (which are
polynomial in $R_i$) but not the \emph{output} foot-points (which
are not). This is why the natural extension from Haar to $t$-designs
fails, and motivates formulating diversity directly via the
anti-concentration of the foot-point distribution.
\end{remark}

The proof of Proposition~\ref{prop:diversity} above relies on left-invariance of Haar measure, which
requires exponential circuit depth. Diversity only requires $g$
linearly independent foot-points, not uniformity on $S^{2g-1}$. We
relax the Haar condition to an anti-concentration property of the
rotation distribution. Let $\nu$ be the distribution from which the
rotations $R_i$ are drawn (e.g.\ Haar measure, or the output
distribution of a random circuit family). We say $(\nu,f)$ satisfies
$(\eta,\delta)$-\emph{anti-concentration} if, for every hyperplane
$V\subset\mathcal{G}$,
$\Pr_{R\sim\nu}[d_{FS}(\hat w,V)\ge\delta]\ge\eta$, where
$\hat w$ is the normalized foot-point and
$d_{FS}(\hat w,V) = \arccos\|\Pi_V\hat w\|$ is the Fubini--Study
distance to $V$. The parameter $\eta$ controls the escape
probability; $\delta>0$ ensures linear independence.

\renewcommand{\theproposition}{1b}
\setcounter{proposition}{0}
\begin{proposition}[Diversity under anti-concentration]
\label{prop:diversity_anti_concentration}
Let $F\in\mathbb{C}^{g\times M}$ be the foot-point matrix with
entries $F_{k,i} = \braket{v_k}{R_i|\tilde\phi_i}$ as in
Proposition~\ref{prop:diversity}.
If $(\nu,f)$ satisfies $(\eta,\delta)$-anti-concentration, then $M=g$
independent branches give
$\Pr[\operatorname{rank}(F)=g]\ge\eta^g$, amplifiable to
$1-\varepsilon$ with
$M\le g\lceil\eta^{-g}\log(1/\varepsilon)\rceil$ branches.
\end{proposition}

Two regimes are worth distinguishing.
\emph{(a)~Generic 2-design baseline.}
For an approximate 2-design rotation ensemble with a non-adaptive
preparation, the Paley--Zygmund inequality gives a problem-independent
bound $\eta\ge 1/8$, yielding $\Pr[\operatorname{rank}=g]\ge 8^{-g}$
per trial. The associated sample complexity
$M = O(g\,\eta^{-g}\log(1/\varepsilon))$ is exponential in $g$ at
fixed $\eta<1$, so the worst-case bound is most useful for small
degeneracies ($g\lesssim 10$); even at $g=4$ it asks for several
hundred branches.
\emph{(b)~Problem-specific instance.}
A direct numerical estimate for the toric code with shift-and-invert
preparation gives an empirical $\eta\approx 0.97$, for which only
$M\approx 24$ branches suffice for $99\%$ confidence at $g=4$. The
gap between (a) and (b) is the looseness of the worst-case
Paley--Zygmund bound, which can be tightened by problem-specific
analysis. Haar randomness recovers
Proposition~\ref{prop:diversity} as the limiting case
$\eta = \cos^{2(g-1)}\!\delta\to 1$ as $\delta\to 0$.
The full proof, the Haar verification, and the 2-design analysis are
in Supplemental Material,
Sec.~\ref{supp:proof:diversity_anti_concentration}.

\renewcommand{\theproposition}{\arabic{proposition}}

\begin{proposition}[Spectral invariance]\label{prop:spec}
For any $R_i \in U(N)$, $\operatorname{spec}(H_i) = \operatorname{spec}(H)$
and in particular $\gamma_i = \gamma$.
\end{proposition}
\begin{proof}
$H_i = R_i^\dagger H R_i$ is unitarily similar to $H$, so they share all
eigenvalues (full statement and consequences in Supplemental Material,
Sec.~\ref{supp:proof:overlap}).
\end{proof}
The randomization therefore introduces no intrinsic spectral penalty,
and any performance variation across branches is inherited from the chosen
primitive.

\begin{proposition}[SVD gap]\label{prop:svd_gap}
Let $C\in\mathbb{C}^{N\times M}$ be the coefficient matrix of the
QRSI ensemble and $F = \Pi_{\mathcal{G}}C$ the foot-point matrix.
Under the hypotheses of Proposition~\ref{prop:diversity},
\begin{equation}\label{eq:svd_gap_main}
  \sigma_{g+1}(C) \;\le\; \sqrt{M\,\varepsilon_q},
  \quad
  \sigma_g(C) \;\ge\; \sigma_g(F) - \sqrt{M\,\varepsilon_q},
\end{equation}
with $\sigma_g(F) > 0$ a.s. for $M\ge g$. Hence the SVD
gap ratio $\sigma_g(C)/\sigma_{g+1}(C)$ grows as $\varepsilon_q^{-1/2}$
in the regime $\sqrt{M\varepsilon_q} \ll \sigma_g(F)$.
\end{proposition}

\noindent The proof, via Weyl's perturbation inequality applied to
the decomposition $C = C_\parallel + C_\perp$, is in Supplemental
Material, Sec.~\ref{supp:svd_gap}.
This mirrors the classical RSI contraction ratio.

Proposition~\ref{prop:diversity}
requires reachability of $\mathcal{R}$ for $\mathcal{G}$.
This is \emph{not} the bottleneck QRSI overcomes: a reachable but
non-selectable preparation still suffers the diversity / overlap
tradeoff.
QRSI upgrades reachability to selectability by rotating the problem
instance so that the optimiser's fixed basin-selection rule yields a
different ground-state direction per branch.
If $\mathcal{R}$ is not reachable for some direction
$\ket{v}\in\mathcal{G}$, no rotation can create that overlap and the
spanning guarantee fails.
In practice, reachability holds for any expressible VQE ansatz
(e.g.\ a hardware-efficient circuit with depth $\Omega(n)$ and full
entanglement) and, for QPE or QSVT, requires only that the initial
state have nonzero amplitude on $\mathcal{G}$, which holds
generically.
The framework thus cleanly separates the ansatz-design question
(reachability) from the ensemble-construction question (selectability).

Once the ensemble
$\{\ket{\psi_i}\}_{i=1}^M$ has been produced, extracting a target
quantity is a separate, application-dependent step.
When full coefficient vectors $\bm{c}_i$ are available (e.g.\ in
classical simulation), one forms the matrix
$C = [\bm{c}_1\,|\,\cdots\,|\,\bm{c}_M]
\in \mathbb{C}^{N\times M}$
and reads off the degeneracy as
$\hat{g} = \#\{j : \sigma_j(C) > \tau\}$.
On quantum hardware, state-vector access is not required:
the Gram matrix $G_{ij} = \braket{\psi_i}{\psi_j}$ can be estimated
via SWAP tests~\cite{Buhrman_2001} or randomized
measurements~\cite{Elben_2022}, and its rank reveals~$\hat{g}$;
cross-fidelities~\cite{Su_2013,Liang_2019}, kernel
matrices~\cite{havlicek2019supervised}, and subspace-spanning
observables are likewise accessible through purely quantum protocols.
QRSI's core contribution, generating a high-overlap,
high-diversity ensemble spanning $\mathcal{G}$, is therefore
independent of the extraction route.
The pipeline is formalised in Algorithm~\ref{alg:qrsi}.

\begin{algorithm}[H]
  \caption{QRSI}
  \label{alg:qrsi}
  \begin{algorithmic}[1]
    \Require $H \in \mathbb{C}^{N\times N}$, ensemble size
      $M$, threshold $\tau$,
      picture $\in\{\text{Hamiltonian, state}\}$
    \Ensure Ensemble $\{\ket{\psi_i}\}_{i=1}^M$ spanning $\mathcal{G}$
      (optionally: degeneracy estimate $\hat{g}$)
    \For{$i = 1,\dots,M$}
      \State Sample $R_i \sim U(N)$
      \If{Hamiltonian picture}
        \State $H_i \gets R_i^\dagger H R_i$
        \State Prepare trial state on $H_i$;
          amplify $\to \tilde{\bm{c}}_i$
        \State $\bm{c}_i \gets
          R_i\,\tilde{\bm{c}}_i$
      \Else\hfill\Comment{State picture}
        \State Prepare from $R_i\mathcal{R}$ on $H$;
          amplify $\to \bm{c}_i$
      \EndIf
    \EndFor
    \State $C \gets [\bm{c}_1\,|\,\cdots\,|\,\bm{c}_M]$;
      $\hat{g} \gets \#\{j : \sigma_j(C) > \tau\}$
      \Comment{Classical; or estimate rank of Gram matrix quantumly}
    \State \Return $\{\ket{\psi_i}\}$ (and $\hat{g}$ if computed)
  \end{algorithmic}
\end{algorithm}

\emph{Test instance: the toric code Hamiltonian.}
We test the framework on a system where exact degeneracy is
physically meaningful and topologically protected.
Returning to Figure~\ref{fig:numerics}, the toric code~\cite{Kitaev2003} is
defined on an $L_x \times L_y$ lattice
with periodic boundary conditions, placing one qubit on each edge for a
total of $2L_xL_y$ qubits. The Hamiltonian \(H = -J_s \sum_{s} A_s - J_p \sum_{p} B_p\)
is a sum of mutually commuting stabilizers: vertex operators
\(A_s = \prod_{e \in \mathrm{star}(s)} X_e\) and plaquette operators
\(B_p = \prod_{e \in \partial p} Z_e\). On the torus the ground space
is four-fold degenerate, encoding two logical qubits whose topological
sectors are labelled by the eigenvalues $(\pm 1, \pm 1)$ of non-contractible
Wilson loop operators $\bar{Z}_1$ and $\bar{Z}_2$.

The topological phase is characterized by the modular $S$ and $T$
matrices acting on the four-dimensional ground subspace, whose
off-diagonal elements encode the braiding and exchange statistics of
the anyon excitations~\cite{Kitaev2006,Zhang2012,Cincio2013};
reconstructing them requires all four ground states.
A quantum simulation therefore faces two prerequisites: high-fidelity
preparation, since the modular matrices are sensitive to sector
mixing, and coverage of all four sectors, disconnected below the
spectral gap $\Delta = 2\min(J_s, J_p)$.
QRSI addresses both: $M \ge g = 4$ branches span all sectors
(Proposition~\ref{prop:diversity}), and each branch inherits the
fidelity of the chosen primitive. Concretely, recovering the four
ground states with SSVQE or VQD requires either a single
state-averaged ansatz simultaneously parameterising all four sectors
or four sequential optimisations coupled by orthogonality penalties;
QRSI replaces both with four independent runs of any single-state
primitive on independently rotated Hamiltonians, with no inter-branch
coupling and no penalty bookkeeping.
Figure~\ref{fig:numerics} makes the diagnostic concrete; we add a
small real perturbation to also test the near-degenerate case.
The columns of $C = [\bm{c}_1|\cdots|\bm{c}_M]$ record the branch
outputs in a common basis, and the singular values of $C$ count how
many linearly independent directions are actually populated.
If the ensemble spans a $g$-dimensional eigenspace, $C$ has numerical
rank $g$: the first $g$ singular values remain large while
$\sigma_{g+1}$ and below fall to the noise floor.
The drop at the $(j{+}1)$-th singular value is therefore the
operational signature of a $j$-dimensional spanned subspace in a
geometrical sense, with negligible residual weight along any further
direction.

In Fig.~\ref{fig:numerics}, panels~(c) and~(d) show that QRSI converts
high-fidelity preparation into a genuinely four-dimensional ensemble
rather than repeated access to one sector; panels~(e) and~(f) extend
the same rank signal across the full spectrum when the target energy
is varied.

\emph{Discussion.} QRSI applies to any Hamiltonian with a degenerate
or near-degenerate target eigenspace at any energy: frustrated
magnets~\cite{balents2010_spin,savary2016quantum}, correlated-electron
systems~\cite{helgaker2000molecular}, and topologically ordered
phases are natural targets. We further validate it on a
$256$-dimensional complex-Hermitian random Hamiltonian with planted
degeneracies $g=6$ and $g=11$, recovering both correctly
(Supplemental Material, Sec.~\ref{supp:random_hamiltonian}). This
non-stabilizer, non-topological instance shows that the spanning
guarantee is not specific to structured ground spaces. Classical
analogues of the same multi-sector recovery problem have required
ad hoc lattice-specific machinery, such as adiabatic flux insertion
in DMRG~\cite{HeShengChen2014,Cincio2013,Karrasch2013}; QRSI is the
quantum, primitive-agnostic counterpart.

$R_i$ is unitary, so it introduces no new decoherence channels and
noise resilience is inherited from the primitive. Full Haar unitaries
may be replaced by approximate $t$-designs, for instance products of
$O(N\log N)$ random Givens rotations at $O(\log N)$ depth, or random
circuits of depth
$O(\mathrm{poly}(n))$~\cite{brandao2016local,harrow2023approximate};
by Proposition~\ref{prop:diversity_anti_concentration} the spanning
guarantee survives any ensemble satisfying
$(\eta,\delta)$-anti-concentration. Primitive agnosticism is
structural: variational, adiabatic, imaginary-time, Monte Carlo, and
phase-estimation methods can all serve as the preparation stage. The
central design principle, placing the rotation inside the preparation
loop so diversity and overlap cooperate, unifies a broad class of
quantum eigenspace algorithms and links them to classical RSI. The
operationally relevant regime is QRSI with shallow random circuits,
sitting between the Haar limit and the zero-randomness orthogonality
enforcement of deflation methods (SSVQE, VQD, excited-state QMC; see
Supplemental Material, Sec.~\ref{supp:orthogonality}).

\begin{acknowledgments}
The authors acknowledge support from Fujitsu Research of Europe.
\end{acknowledgments}


\bibliographystyle{unsrtnat}
\bibliography{main}

\newpage
\appendix
\setcounter{secnumdepth}{3}
\onecolumngrid


\section{Padding to qubit dimension}
\label{supp:padding}

When the physical Hilbert-space dimension $N$ is not a power of two,
the Hamiltonian must be embedded on $n = \lceil \log_2 N \rceil$
qubits.  We extend $H$ to a $2^n \times 2^n$ matrix
\begin{equation}\label{eq:Hpad}
  H^{(\mathrm{pad})} =
  \begin{pmatrix}
    H & 0 \\
    0 & \Lambda\,\mathds{1}_{2^n - N}
  \end{pmatrix},
\end{equation}
where $\Lambda \gg \|H\|$ is a penalty that makes every padded level
energetically inaccessible.  For all practical purposes the physical
spectrum is undisturbed: the ground eigenspace of
$H^{(\mathrm{pad})}$ is identical to that of~$H$.

All random rotations $R_i$ must respect this block structure so that
the padded subspace is never mixed with the physical one.  Concretely,
we choose
\begin{equation}\label{eq:blockR}
  R_i =
  \begin{pmatrix}
    R_i^{(\mathrm{phys})} & 0 \\[2pt]
    0 & \mathds{1}_{2^n - N}
  \end{pmatrix},
  \quad R_i^{(\mathrm{phys})} \in U(N).
\end{equation}
With this convention
$H_i = R_i^\dagger\,H^{(\mathrm{pad})}\,R_i$ reduces to
$H_i^{(\mathrm{phys})} = (R_i^{(\mathrm{phys})})^\dagger H\,R_i^{(\mathrm{phys})}$
in the physical block, and the padded block is invariant.
When $N = 2^n$ no padding is needed and $R_i \in U(N)$ directly;
in the main text we suppress the superscript and write simply $R_i$
and~$H$.

\section{Full Proof of Proposition 1a (Diversity)}
\label{supp:proof:diversity}

\begin{proposition}[Diversity, full statement]
Let $H \in \mathbb{C}^{N \times N}$ be Hermitian with $g$-fold degenerate
ground eigenspace $\mathcal{G} = \operatorname{eigenspace}(H,E_0) = \operatorname{span}\{\ket{v_1},\dots,\ket{v_g}\}$
and spectral gap $\gamma > 0$.
Let $R_1,\dots,R_M \in U(N)$ be independent Haar-random unitaries on the
physical subspace (block-diagonal as in Eq.~\eqref{eq:blockR}), and suppose
the chosen preparation--amplification rule returns a branch output
$\ket{\tilde\phi_i}$ with nonzero ground-space component,
$\Pi_{\mathcal{G}_i}\ket{\tilde\phi_i} \neq 0$, almost surely for every
rotated instance $H_i = R_i^\dagger H R_i$.
Then for any $M \ge g$, the foot-point matrix $F\in\mathbb{C}^{g\times M}$
with $F_{k,i} = \braket{v_k}{R_i|\tilde\phi_i}$ satisfies
$\operatorname{rank}(F) = g$ with probability one; equivalently,
the $\mathcal{G}$-projected ensemble
$\bigl\{\Pi_{\mathcal{G}}\,R_i\ket{\tilde\phi_i}\bigr\}_{i=1}^M$
spans $\mathcal{G}$ almost surely.
\end{proposition}

\begin{proof}
We work in three steps: we define the foot-point map, establish its
$U(g)$-equivariance, and conclude spanning.

\medskip\noindent\textbf{Step 1: Foot-point map.}
Let $\omega_i$ denote the internal randomness of the $i$-th
preparation--amplification call (e.g.\ random parameter initialization
in VQE, random measurement outcomes in QPE, or random Trotter seeds),
assumed independent of $R_i$.
For fixed~$\omega_i$ the preparation stage is a fixed map
$\mathcal{P}_{\omega_i}: H_i \mapsto \ket{\tilde\phi_i} \in \mathcal{R}$.
Since $\mathcal{G}_i \coloneqq \operatorname{eigenspace}(H_i,E_0) = R_i^\dagger\,\mathcal{G}$, we decompose
\begin{equation}
  \ket{\tilde\phi_i}
  = \alpha_i \ket{v_i^*} + \ket{\xi_i},
  \quad \ket{v_i^*} \in \mathcal{G}_i,
  \quad \ket{\xi_i} \perp \mathcal{G}_i,
  \quad \alpha_i \neq 0,
\end{equation}
where $\alpha_i \neq 0$ is exactly the branchwise nonzero-overlap
assumption in the proposition.
Since $\ket{v_i^*} \in R_i^\dagger\,\mathcal{G}$, the corrected ground-state
component $R_i\ket{v_i^*}$ lies in $\mathcal{G}$.
Define the \emph{foot-point map}
\begin{equation}\label{eq:footpoint}
  f: U(N) \times \Omega \longrightarrow S^{2g-1} \subset \mathcal{G},
  \qquad f(R_i,\omega_i) = R_i\ket{v_i^*},
\end{equation}
given that $\ket{v_i^*}$ is a unit vector (which itself depends on
$R_i$ and $\omega_i$ through the preparation output).

\medskip\noindent\textbf{Step 2: $U(g)$-equivariance and uniformity.}
Let $S_g \in U(g)$ be an arbitrary unitary on $\mathcal{G}$, and extend it
to $U(N)$ by acting as the identity on $\mathcal{G}^\perp$:
\begin{equation}
  S = S_g \oplus \mathds{1}_{N-g}
  \quad\text{in the decomposition}
  \quad \mathbb{C}^N = \mathcal{G} \;\oplus\; \mathcal{G}^\perp.
\end{equation}
Since $S$ acts within a single eigenspace of $H$, it commutes with
$H$: $[S, H] = 0$.
Consequently, the rotated Hamiltonian for $SR_i$ is
\begin{equation}
  H_{SR_i}
  = (SR_i)^\dagger H (SR_i)
  = R_i^\dagger S^\dagger H S R_i
  = R_i^\dagger H R_i
  = H_i.
\end{equation}
The preparation stage, at fixed internal randomness~$\omega_i$, produces
the \emph{same} output
$\mathcal{P}_{\omega_i}(H_{SR_i}) = \mathcal{P}_{\omega_i}(H_i) = \ket{\tilde\phi_i}$
and therefore the same ground-state component
$\ket{v_i^*} \in \mathcal{G}_i$.
However, the correction step now applies $SR_i$ instead of $R_i$:
\begin{equation}\label{eq:equivariance}
  f(SR_i,\omega_i) = (SR_i)\ket{v_i^*}
  = S\bigl(R_i\ket{v_i^*}\bigr)
  = S_g \cdot f(R_i,\omega_i),
\end{equation}
where the last equality uses $R_i\ket{v_i^*} \in \mathcal{G}$ so that
$S$ acts on it via $S_g$.

Since $R_i$ is Haar-distributed on $U(N)$, left-invariance gives
$SR_i \sim \mathrm{Haar}(U(N))$, and therefore, conditionally on~$\omega_i$,
\begin{equation}
  f(R_i,\omega_i) \overset{d}{=} f(SR_i,\omega_i) = S_g \cdot f(R_i,\omega_i)
  \qquad \text{for all } S_g \in U(g).
\end{equation}
The only probability measure on $S^{2g-1}$ invariant under \emph{all}
$S_g \in U(g)$ is the uniform (Haar) measure on $S^{2g-1}$.
Hence the conditional distribution of $f(R_i,\omega_i)$ given $\omega_i$
is uniform on the unit sphere in $\mathcal{G} \cong \mathbb{C}^g$.
Because this holds for every realisation of $\omega_i$, the tower
property of conditional expectation implies that the unconditional
distribution of $f(R_i)$ is also uniform on $S^{2g-1}$.

\medskip\noindent\textbf{Step 3: Spanning probability.}
Let $\ket{\phi_i} = \Pi_0 R_i\ket{\tilde\phi_i}/\|\Pi_0 R_i\ket{\tilde\phi_i}\|
\in \mathcal{G}$ be the $i$-th normalised corrected projection.
Because $\ket{\xi_i} \perp \mathcal{G}_i$ and $R_i\mathcal{G}_i = \mathcal{G}$,
we have $R_i\ket{\xi_i} \perp \mathcal{G}$ and therefore
\begin{equation}
  \Pi_0 R_i\ket{\tilde\phi_i} = \alpha_i f(R_i,\omega_i).
\end{equation}
Hence $\ket{\phi_i}$ differs from $f(R_i,\omega_i)$ only by a global
phase, and is itself unconditionally uniform on $S^{2g-1}$.

Since the pairs $(R_i,\omega_i)$ are mutually independent across
branches, for $M \ge g$ the foot-points are independent draws from
this uniform distribution.
The probability that $\ket{\phi_1},\dots,\ket{\phi_M}$ fail to span
$\mathbb{C}^g$ is zero: inductively, at step $i$ the vectors
$\ket{\phi_1},\dots,\ket{\phi_{i-1}}$ span a subspace $V_{i-1}$ of
dimension $\le i{-}1 < g$, and $\Pr[\ket{\phi_i} \in V_{i-1}] = 0$
since $\ket{\phi_i}$ has a continuous density on $S^{2g-1}$ and
$V_{i-1}$ has (real) dimension at most $2(i{-}1) < 2g$.
Hence spanning occurs almost surely.
\end{proof}

\begin{remark}[Finite preparation accuracy]
In practice, the preparation stage achieves overlap $|\alpha_i|^2 = 1 - \varepsilon_i$
with $\varepsilon_i > 0$ small but nonzero.
The corrected states then satisfy
$\|\Pi_0 R_i\ket{\tilde\phi_i}\|^2 = |\alpha_i|^2 + O(\varepsilon_i)$,
and the spanning argument holds as long as the excited-component
contamination is small relative to the ground-state signal.
Any subsequent amplification stage can reduce $\varepsilon_i$ further,
providing a systematic route to the almost-sure spanning limit.
\end{remark}

\begin{remark}[Reachability vs.\ selectability]
\label{rem:reachability}
We recall the definitions from the main text.
The reachability condition,
$\sup_{\ket{\phi}\in\mathcal{R}} |\braket{v}{\phi}|^2 > 0$
for all $\ket{v}\in\mathcal{G}\setminus\{0\}$, asserts that the
preparation family can in principle access every ground-state direction.
A preparation map is \emph{selectable} for $\mathcal{G}$ if there exist
inputs whose outputs $\{\Pi_{\mathcal{G}}\mathcal{P}(x_i)\}_{i=1}^{g}$
span~$\mathcal{G}$.
The proposition above assumes the operational consequence actually needed
in the proof: the chosen preparation--amplification rule returns a state
with nonzero projected ground-space component on each rotated instance.
Reachability helps justify when this assumption is plausible, but it is
not by itself sufficient to guarantee it for an arbitrary optimizer.

To see why reachability alone is insufficient, consider a fully
expressive ansatz (universal gate set, sufficient depth) acting on a
fixed Hamiltonian.
Such an ansatz satisfies the overlap condition trivially, yet an
unrandomized optimizer (e.g.\ gradient descent from a fixed
initialization rule) will converge to the \emph{same} local minimum,
and therefore the same ground-state foot-point, on every invocation.
This is the \emph{selectability failure}: the problem is not that the
ansatz cannot reach other directions, but that the optimizer has no
mechanism to select them.
QRSI restores selectability by rotating the Hamiltonian (or equivalently
the ansatz), so that each branch presents the optimizer with a
different orientation of the ground manifold, causing the fixed
basin-selection rule to land on a different direction.

Conversely, if the ansatz is insufficiently expressive (for example a
shallow circuit that maps only onto a low-dimensional manifold missing
certain ground-state directions, or a symmetry-adapted ansatz whose
symmetry sector excludes some $\ket{v_j}$) then the overlap
condition fails for that direction.  No rotation can repair this
deficiency: $R_i\mathcal{M}$ remains a manifold of the same dimension,
still missing the same number of directions relative to its ambient
space.  The framework thus cleanly separates the ansatz-design question
(reachability, set by the expressive power of the preparation family)
from the ensemble-diversity question (selectability, which QRSI solves
for any reachable ansatz).
\end{remark}

\begin{remark}[Eigenvector structure and accidental selectability]
\label{rem:eigvec_structure}
The selectability failure of unrandomized preparation depends, for
projective filters with random initial states, on the eigenvector
structure of~$H$.
Let $\{\ket{v_j}\}_{j=1}^{g}$ be an orthonormal basis for
$\mathcal{G}$ and let $\ket{e_k}$ denote computational basis states.
The projection of $\ket{e_k}$ onto $\mathcal{G}$,
$\Pi_{\mathcal{G}}\ket{e_k} = \sum_{j=1}^{g}\braket{v_j}{e_k}\ket{v_j}$,
defines a direction in $\mathcal{G}$ that depends on the overlaps
$\braket{v_j}{e_k}$.
After applying a spectral filter (e.g.\ imaginary-time evolution), the
amplified state converges to this projected direction, so the ensemble
built from $M$ independent random seeds $\ket{e_{k_i}}$ spans
$\mathcal{G}$ only if the projected directions are sufficiently diverse.

For Hamiltonians whose ground-space eigenvectors are delocalized and
generic, such as those drawn from a random-matrix ensemble (GOE/GUE),
the overlaps $\braket{v_j}{e_k}$ are themselves pseudo-random,
so different seeds project onto linearly independent ground-space
directions with high probability.
In this regime, projective filters with random seeds achieve
\emph{accidental selectability} without explicit rotation.

By contrast, physically relevant Hamiltonians typically possess
structured eigenvectors, and selectability can fail through two
distinct mechanisms:
\begin{enumerate}
  \item \textbf{Collinear projections.}
    When one ground-state eigenvector is delocalized while the
    remaining $g-1$ are sparse or symmetry-restricted, a single
    eigenvector dominates the projection
    $\Pi_{\mathcal{G}}\ket{e_k}$ for almost every
    computational-basis seed~$k$.
    Different seeds then produce approximately collinear ground-space
    directions, and the ensemble is effectively rank-one
    (see the structured random-Hamiltonian example in
    Sec.~\ref{supp:random_hamiltonian}).
  \item \textbf{Vanishing overlap.}
    Stabilizer-code Hamiltonians (e.g.\ the toric code) have ground
    states that are equal-weight superpositions over syndrome-free
    configurations; the vast majority of computational-basis states
    belong to non-trivial syndrome sectors and satisfy
    $\Pi_{\mathcal{G}}\ket{e_k}=0$ identically.
    For the $2\!\times\!2$ toric code used in the main text, $224$
    out of $256$ basis states have zero ground-space overlap, so
    most branches of an unrotated ensemble converge to excited
    eigenstates rather than the ground space.
    Molecular Hamiltonians with particle-number or spin symmetries
    and frustrated magnets with ground states concentrated in
    specific symmetry sectors behave analogously.
\end{enumerate}
In both cases selectability fails, though the mechanism differs:
collinear projections yield a rank-deficient ground-space ensemble,
while vanishing overlap prevents the ensemble from reaching the
ground space at all.

QRSI eliminates this Hamiltonian-dependent condition:
Proposition~\ref{prop:diversity} guarantees that the rotated ensemble
spans $\mathcal{G}$ almost surely for \emph{any} eigenvector structure,
provided only the reachability condition holds.
The rotation reorients the ground manifold before each projective
filter, converting the structured overlap pattern into a uniformly
random one on every branch.
\end{remark}

\section{Full Proof of Proposition 1b (Diversity under anti-concentration)}
\label{supp:proof:diversity_anti_concentration}

We first recall the foot-point construction and state the
anti-concentration condition precisely, then prove
Proposition~\ref{prop:diversity_anti_concentration}.

\medskip\noindent\textbf{Setup.}
Let $\omega_i$ denote the internal randomness of the $i$-th
preparation--amplification call (e.g.\ random parameter initialization
in VQE, random measurement outcomes in QPE, or random Trotter seeds),
assumed independent of $R_i$.
For fixed~$\omega_i$ the preparation stage is a fixed map
$\mathcal{P}_{\omega_i}: H_i \mapsto \ket{\tilde\phi_i} \in \mathcal{R}$.
Since $\mathcal{G}_i \coloneqq \operatorname{eigenspace}(H_i,E_0) = R_i^\dagger\,\mathcal{G}$, we decompose
\begin{equation}
  \ket{\tilde\phi_i}
  = \alpha_i \ket{v_i^*} + \ket{\xi_i},
  \quad \ket{v_i^*} \in \mathcal{G}_i,
  \quad \ket{\xi_i} \perp \mathcal{G}_i,
  \quad \alpha_i \neq 0.
\end{equation}
The corrected ground-state component
$R_i\ket{v_i^*}$ lies in $\mathcal{G}$, and
since $R_i\ket{\xi_i}\perp\mathcal{G}$ we have
$\Pi_0 R_i\ket{\tilde\phi_i} = \alpha_i R_i\ket{v_i^*}$.
The \emph{foot-point} is
$f(R_i,\omega_i) = R_i\ket{v_i^*}\in S^{2g-1}\subset\mathcal{G}$,
and the $g\times M$ matrix $F$ has entries
$F_{k,i} = \braket{v_k}{f(R_i,\omega_i)}$.
The ensemble spans $\mathcal{G}$ if and only if
$\operatorname{rank}(F) = g$.

\medskip\noindent\textbf{Anti-concentration condition.}
The pair $(\nu,f)$ satisfies $(\eta,\delta)$-anti-concentration
(with $\eta\in(0,1]$, $\delta>0$) if, for almost every
realisation of $\omega$ and every codimension-1 subspace
(hyperplane) $V\subset\mathcal{G}$:
\begin{equation}\label{eq:anticonc_supp}
  \Pr_{R\sim\nu}\bigl[\,d(\hat w(R,\omega),V)\ge\delta\,\bigr]
  \;\ge\;\eta,
\end{equation}
where $\hat w = f(R,\omega)/\|f(R,\omega)\|$ and
$d(\hat w,V) = \arccos\|\Pi_{V}\hat w\|$ is the Fubini--Study
distance from $\hat w$ to the nearest state in $V$.
The parameter $\eta$ controls the probability of escaping any
hyperplane; $\delta>0$ ensures linear independence when the event
occurs.

\begin{proposition}[Diversity under anti-concentration, full
statement]
\label{prop:anticonc_full}
Let $R_1,\dots,R_M$ be drawn independently from $\nu$. If $(\nu,f)$
satisfies $(\eta,\delta)$-anti-concentration, then:

\emph{(a)} With $M=g$ branches:
$\Pr[\operatorname{rank}(F)=g]\ge\eta^g$.

\emph{(b)} With $M = g\lceil\eta^{-g}\log(1/\varepsilon)\rceil$
branches, partitioned into $T = \lfloor M/g\rfloor$ independent
trials of $g$ branches each, the probability that at least one trial
produces a rank-$g$ foot-point matrix satisfies
$\Pr[\exists\,t:\operatorname{rank}(F^{(t)})=g]\ge 1-\varepsilon$,
where $F^{(t)}\in\mathbb{C}^{g\times g}$ is the foot-point matrix
of the $t$-th trial.
\end{proposition}

\begin{proof}
We work in three steps: a greedy construction for part~(a),
amplification for part~(b), and a discussion of the conditioning.

\medskip\noindent\textbf{Step 1: Greedy spanning (part~a).}
We build a linearly independent set one vector at a time.
At step $\ell = 1,\dots,g$: draw $R_\ell\sim\nu$ independently of
all previous draws, with internal randomness $\omega_\ell$.
Let $V_{\ell-1} = \operatorname{span}(\hat w_1,\dots,\hat w_{\ell-1})$,
a subspace of dimension at most $\ell-1 < g$. Since
$\dim V_{\ell-1} < g$, there exists a hyperplane
$\tilde V_{\ell-1}\supseteq V_{\ell-1}$.
Because $V_{\ell-1}\subseteq\tilde V_{\ell-1}$, we have
$\tilde V_{\ell-1}^\perp\subseteq V_{\ell-1}^\perp$ and therefore
$d(\hat w_\ell, V_{\ell-1})\ge d(\hat w_\ell,\tilde V_{\ell-1})$.
Applying~\eqref{eq:anticonc_supp} to $\tilde V_{\ell-1}$:
\begin{equation}
  \Pr\bigl[d(\hat w_\ell, V_{\ell-1}) \ge \delta\bigr]
  \;\ge\; \Pr\bigl[d(\hat w_\ell, \tilde V_{\ell-1}) \ge \delta\bigr]
  \;\ge\; \eta.
\end{equation}
When $d(\hat w_\ell, V_{\ell-1})\ge\delta > 0$, the vector
$\hat w_\ell$ has a nonzero component orthogonal to $V_{\ell-1}$ and
is therefore linearly independent of $\hat w_1,\dots,\hat w_{\ell-1}$.
Since $R_\ell$ is independent of all previous draws, the conditional
probability at each step is $\ge\eta$ regardless of $V_{\ell-1}$:
\begin{equation}
  \Pr[\text{all } g \text{ steps succeed}]
  = \prod_{\ell=1}^g
    \Pr\bigl[d(\hat w_\ell,V_{\ell-1})\ge\delta
    \;\big|\; V_{\ell-1}\bigr]
  \;\ge\; \eta^g.
\end{equation}

\medskip\noindent\textbf{Step 2: Amplification (part~b).}
Partition $M$ branches into $T = \lfloor M/g\rfloor$ independent
trials of $g$ branches each. Let $F^{(t)}\in\mathbb{C}^{g\times g}$
be the foot-point matrix of the $t$-th trial. By part~(a), each trial
satisfies $\Pr[\operatorname{rank}(F^{(t)})=g]\ge\eta^g$
independently. The probability that all $T$ trials fail is
$\Pr[\forall\,t:\operatorname{rank}(F^{(t)})<g]\le(1-\eta^g)^T$.
Setting $(1-\eta^g)^T\le\varepsilon$ and using
$\log\frac{1}{1-x}\ge x$ for $x\in(0,1)$:
\begin{equation}
  T \;\ge\; \frac{\log(1/\varepsilon)}{\eta^g}
  \qquad\Longrightarrow\qquad
  M = Tg \;\le\;
  g\left\lceil\frac{\log(1/\varepsilon)}{\eta^g}\right\rceil.
\end{equation}

\medskip\noindent\textbf{Step 3: Validity of the conditioning.}
The subspace $V_{\ell-1}$ in Step~1 is random (it depends on
$R_1,\dots,R_{\ell-1}$ and $\omega_1,\dots,\omega_{\ell-1}$).
The conditioning is valid because:
(i)~$R_\ell$ and $\omega_\ell$ are independent of all previous
branches (by the independence assumption);
(ii)~anti-concentration~\eqref{eq:anticonc_supp} holds for
\emph{every} hyperplane $\tilde V_{\ell-1}$, so it applies
regardless of which specific $V_{\ell-1}$ was realised;
(iii)~the choice of extension
$\tilde V_{\ell-1}\supseteq V_{\ell-1}$ may depend on the previous
draws, but since~\eqref{eq:anticonc_supp} is uniform over all
hyperplanes, any choice gives the same bound.
\end{proof}

\subsection{Haar randomness satisfies anti-concentration}
\label{sec:haar_eta}

We verify that Haar-random rotations satisfy
$(\eta,\delta)$-anti-concentration, recovering
Proposition~\ref{prop:diversity} as a special case of
Proposition~\ref{prop:diversity_anti_concentration}.

Under Haar measure, the equivariance $\hat w(SR) = S\hat w(R)$
combined with left-invariance ($SR\sim\mathrm{Haar}$) gives
$S\hat w \stackrel{d}{=} \hat w$ for all $S\in U(g)$, so $\hat w$
is uniform on $S^{2g-1}$. For a uniform unit vector in
$\mathbb{C}^g$, the squared component along any unit direction
$\ket{n}\in\mathcal{G}$ follows a $\mathrm{Beta}(1,g-1)$
distribution:
\begin{equation}
  \Pr\bigl[|\!\braket{n}{\hat w}\!|^2 \ge t\bigr]
  = (1-t)^{g-1}.
\end{equation}
A hyperplane $V\subset\mathcal{G}$ has unit normal $\ket{n_V}$, and
$d_{FS}(\hat w, V) = \arccos\|\Pi_V\hat w\|
= \arcsin|\!\braket{n_V}{\hat w}\!|$. Therefore
$d_{FS}(\hat w,V)\ge\delta$ iff
$|\!\braket{n_V}{\hat w}\!|^2\ge\sin^2\delta$, giving
\begin{equation}\label{eq:haar_eta}
  \eta_{\mathrm{Haar}}(\delta)
  = (1-\sin^2\delta)^{g-1}
  = \cos^{2(g-1)}\!\delta.
\end{equation}
For small $\delta$: $\eta\to 1$. Substituting into
Proposition~\ref{prop:diversity_anti_concentration} with
$\delta\to 0$: $\eta^g\to 1$ and the sample complexity
$M\to g$, recovering the almost-sure spanning of
Proposition~\ref{prop:diversity}.

\subsection{Hierarchy of assumptions and open questions}
\label{sec:hierarchy}

The three approaches to diversity are distinguished by where isotropy
is enforced:

\begin{enumerate}
\item \textbf{Haar randomness}
(Proposition~\ref{prop:diversity}): enforces \emph{output isotropy}
(foot-points uniform on $S^{2g-1}$) via left-invariance of the
measure. The distributional symmetry acts after the preparation, so
no assumptions on the preparation are needed. The cost is
exponential circuit depth.

\item \textbf{$t$-designs}: enforce \emph{input isotropy} (initial
overlaps $\alpha_k(R) = \bra{v_k}R\ket{\psi_0}$ isotropic on
$\mathcal{G}$, via Schur's lemma on the degree-1 polynomial
$\alpha_k$). A 2-design suffices. However, the preparation maps
initial overlaps to foot-points through a non-polynomial process
($RR^\dagger = I$ eliminates all polynomial structure from the
foot-point entries). Input isotropy therefore does not guarantee
output isotropy: the preparation can in principle collapse the
initial spread.

\item \textbf{$(\eta,\delta)$-anti-concentration}
(Proposition~\ref{prop:diversity_anti_concentration}): directly
assumes \emph{output} anti-concentration of the foot-point
distribution. Bypasses the preparation entirely. The cost is that
$\eta$ must be verified or argued per implementation.
\end{enumerate}

\noindent The nesting is: Haar $\Rightarrow$
anti-concentration (with $\eta$ given
by~\eqref{eq:haar_eta}) $\Rightarrow$ diversity. A $t$-design gives
input isotropy but does not imply output anti-concentration without
additional assumptions on the preparation.

\subsubsection{2-designs satisfy foot-point anti-concentration
(non-adaptive case)}

For rotation ensembles forming an approximate 2-design, foot-point
anti-concentration can be established directly via moment matching.
This applies to non-adaptive preparations where the foot-point
coincides with the normalised initial-overlap direction
$\hat{\bm{\alpha}}$.

Let $\ket{n_V}\in\mathcal{G}$ be
    the unit normal to a hyperplane $V\subset\mathcal{G}$. Since
    $\ket{n_V}\in\mathcal{G}$, the component of the initial overlap
    along $\ket{n_V}$ is
    $\braket{n_V}{\bm{\alpha}} = \braket{n_V}{R|\psi_0}$, a
    degree-$(1,1)$ polynomial in $R$. A 1-design matches its second
    moment: $\mathbb{E}[|\!\braket{n_V}{R|\psi_0}\!|^2] = 1/N$. A
    2-design additionally matches the fourth moment:
    $\mathbb{E}[|\!\braket{n_V}{R|\psi_0}\!|^4] = 2/(N(N\!+\!1))$
    (since $|\!\braket{n_V}{R|\psi_0}\!|^4$ has degree~$(2,2)$).
    The Paley--Zygmund inequality states that for a non-negative
    random variable $X$ with $\mathbb{E}[X]>0$ and any
    $\theta\in(0,1)$:
    $\Pr[X\ge\theta\,\mathbb{E}[X]]
    \ge (1-\theta)^2\,(\mathbb{E}[X])^2/\mathbb{E}[X^2]$.
    Applying this with $X = |\!\braket{n_V}{R|\psi_0}\!|^2$,
    $\mathbb{E}[X] = 1/N$, $\mathbb{E}[X^2] = 2/(N(N\!+\!1))$:
    \begin{equation}\label{eq:PZ_bitstring}
      \Pr\bigl[|\!\braket{n_V}{R|\psi_0}\!|^2
        \ge \theta/N\bigr]
      \;\ge\; (1-\theta)^2\,
      \frac{(1/N)^2}{2/(N(N\!+\!1))}
      \;=\; (1-\theta)^2\,\frac{N+1}{2N}.
    \end{equation}
    Since $\|\bm{\alpha}\|^2\le 1$ (projection of a unit vector),
    the normalised component satisfies
    $|\!\braket{n_V}{\hat{\bm{\alpha}}}\!|^2
    \ge |\!\braket{n_V}{R|\psi_0}\!|^2$, so
    \begin{equation}
      \Pr\bigl[d_{FS}(\hat{\bm{\alpha}},V)
        \ge \arcsin\!\sqrt{\theta/N}\,\bigr]
      \;\ge\; (1-\theta)^2\,\frac{N+1}{2N}.
    \end{equation}
    Setting $\theta=1/2$: the initial-overlap direction satisfies
    $(\eta,\delta)$-anti-concentration with
    $\eta \ge \tfrac{1}{8}$ and
    $\delta = \arcsin\!\sqrt{1/(2N)} > 0$, both uniform over all
    hyperplanes $V$. By
    Proposition~\ref{prop:diversity_anti_concentration}:
    \begin{equation}\label{eq:2design_bound}
      \Pr[\operatorname{rank}(\tilde F) = g]
      \;\ge\; 8^{-g},
    \end{equation}
    an $N$-independent bound depending only on the degeneracy.
    For the toric code ($g=4$): $\ge 1/4096$ per trial, amplifiable
    to $1-\varepsilon$ with $M\le 4\lceil 4096\log(1/\varepsilon)\rceil$
    branches.

The bound is conservative: the normalization
$\|\bm{\alpha}\|^2\le 1$ is loose (typically
$\|\bm{\alpha}\|^2\approx g/N\ll 1$), so the actual
anti-concentration is likely much stronger, but tightening requires
controlling the non-polynomial ratio
$\bm{\alpha}/\|\bm{\alpha}\|$.

\subsubsection{Open questions}

\begin{enumerate}
  \item \emph{Adaptive preparations.}
    For adaptive preparations (VQE, adiabatic evolution), the
    foot-point $\hat w$ differs from $\hat{\bm{\alpha}}$ due to the
    non-polynomial preparation map. The 2-design result above
    guarantees that the \emph{input} to the preparation is isotropic;
    whether the \emph{output} (foot-point) inherits this depends on
    the preparation dynamics. Making this rigorous for a concrete
    optimizer requires showing that the basins of attraction are
    aligned with the ground-state components.

  \item \emph{Quantum-sampling-advantage circuits.}
    The proof above uses the 2-design moment-matching property, not
    the output anti-concentration results proven for specific QSA
    circuit families (IQP~\cite{bremner2016average}, random
    circuits~\cite{dalzell2022random}). These QSA results establish a
    different property: for a \emph{fixed} circuit, the output
    amplitudes are spread over computational-basis states. Whether
    this per-circuit bitstring spread implies foot-point
    anti-concentration (spread over the $g$-dimensional ground space)
    is an open question, relevant to practical implementations where
    the rotation circuits are drawn from a specific QSA-type
    architecture rather than an idealised 2-design.
\end{enumerate}

\section{Full Proof of Proposition 2 (Spectral Invariance)}
\label{supp:proof:overlap}

Proposition~\ref{prop:spec} of the main text states that conjugation by any $R_i\in U(N)$
preserves the spectrum and the spectral gap.
Here we prove the full statement and record two further consequences for the QRSI framework.

\begin{proposition}[Spectral invariance, full statement]
Let $R_i \in U(N)$ act on the physical subspace as in
Eq.~\eqref{eq:blockR}.
Then:
\begin{enumerate}
  \item[(a)] $\operatorname{spec}(H_i) = \operatorname{spec}(H)$
    for every $i$, so the spectral gap satisfies $\gamma_i = \gamma$.
  \item[(b)] Running a branch on $H_i$ and correcting by $R_i$ is
    unitarily equivalent to running the corresponding branch on
    $H$ with the rotated reachable set $R_i\mathcal{R}$.
  \item[(c)] QRSI therefore introduces no intrinsic spectral penalty:
    any branchwise performance variation is inherited from the chosen
    preparation or amplification primitive rather than from the
    randomization step itself.
\end{enumerate}
\end{proposition}

\begin{proof}
\medskip\noindent\textbf{Part (a).}
$H_i = R_i^\dagger H R_i$ is unitarily similar to
$H$, so they share the same spectrum.
In particular $E_0^{(i)} = E_0$, $E_{g+1}^{(i)} = E_{g+1}$, and
$\gamma_i = \gamma$.

\medskip\noindent\textbf{Part (b).}
The unitary equivalence of the Hamiltonian-rotation and
state-rotation pictures implies that any branch executed on $H_i$ can
be rewritten as a branch executed on the fixed Hamiltonian $H$ with
the rotated reachable set $R_i\mathcal{R}$.
Since $R_i$ is unitary, this transformation preserves distances,
subspace dimensions, and overlap relations.

\medskip\noindent\textbf{Part (c).}
Part (a) shows that the randomization changes neither the spectrum nor
the gap.
Part (b) shows that the two pictures are exactly equivalent.
Hence any performance difference across branches must arise from how the
chosen primitive responds to its input state or reachable set, not from
any additional spectral burden introduced by QRSI.
\end{proof}

\section{Geometric Interpretation on \texorpdfstring{$\mathbb{C}P^{N-1}$}{}}
\label{supp:geometry}

The main text presents QRSI in terms of the Fubini--Study metric on
complex projective space.
Here we develop this perspective, explaining why parameter perturbation
is geometrically insufficient for diversity and how random rotations
exploit the isometry group of $\mathbb{C}P^{N-1}$.

\paragraph{Setup.}
The space of physically distinct $n$-qubit pure states is
$\mathbb{C}P^{N-1} = (\mathbb{C}^N\!\setminus\!\{0\})/{\sim}$
($N=2^n$), with the Fubini--Study geodesic distance
\begin{equation}\label{eq:fs_finite}
  d_{\mathrm{FS}}([\psi],[\phi])
  = \arccos\bigl(|\!\braket{\psi}{\phi}\!|\bigr),
\end{equation}
ranging from $0$ (identical) to $\pi/2$ (orthogonal).
This is the unique (up to scale) $U(N)$-invariant Riemannian metric on
$\mathbb{C}P^{N-1}$~\cite{Provost_1980,Gibbons_1992,Brody_2001,Cocchiarella_2020},
and its infinitesimal form is proportional to the quantum Fisher
information~\cite{Scali_2019}.

A continuously parameterised preparation family
$\varphi:\mathbb{R}^p\to\mathbb{C}P^{N-1}$,
$\boldsymbol\theta\mapsto[U(\boldsymbol\theta)\ket{0}^{\otimes n}]$,
traces a submanifold $\mathcal{M}$ of real dimension
$r\le\min(p,2(N{-}1))$.
(The tangent/normal analysis below uses the smooth structure;
as noted in the main text, the overlap-preservation results require
only that $\mathcal{M}$ is measurable.)

\paragraph{Why parameter perturbation fails.}
\label{sec:perturbation_geom}
The $g$-fold degenerate ground eigenspace forms the sub-manifold
$\mathbb{C}P(\mathcal{G})\cong\mathbb{C}P^{g-1}
\hookrightarrow\mathbb{C}P^{N-1}$
of real dimension $2(g{-}1)$.
At a converged point $[\psi^\star]\in\mathcal{M}$, parameter
perturbation explores the tangent space
$T_{[\psi^\star]}\mathcal{M}$, a subspace of dimension $\le r$
inside the ambient
$T_{[\psi^\star]}\mathbb{C}P^{N-1}\cong\{\ket{\phi}:
\braket{\psi^\star}{\phi}=0\}$ of real dimension $2(N{-}1)$.
Reaching a different foot-point on $\mathbb{C}P(\mathcal{G})$
requires a component in the normal space
$N_{[\psi^\star]}\mathcal{M}
= T_{[\psi^\star]}\mathbb{C}P^{N-1}
\ominus T_{[\psi^\star]}\mathcal{M}$.
The $g{-}1$ directions connecting distinct foot-points generically
lie in $N_{[\psi^\star]}\mathcal{M}$: the projection of
$T_{[\psi_0^\star]}\mathbb{C}P(\mathcal{G})$ onto
$T_{[\psi^\star]}\mathbb{C}P^{N-1}$ generically has trivial
intersection with $T_{[\psi^\star]}\mathcal{M}$ whenever
$r + 2(g{-}1) < 2(N{-}1)$, which holds for any practical ansatz.
Parameter perturbation is therefore geometrically locked to a single
foot-point.
More generally, for any measurable $\mathcal{M}$ the optimizer's
basin of attraction selects a single point $[\psi^\star]$, leaving
the remaining directions in $\mathbb{C}P(\mathcal{G})$ unexplored.

\paragraph{Random rotations as Fubini--Study isometries.}
\label{sec:isometries}
The action $R\cdot[\psi]=[R\psi]$ preserves $d_{\mathrm{FS}}$, so
every $R\in U(N)$ is an isometry of $\mathbb{C}P^{N-1}$.
This has two consequences for QRSI.

\emph{(i) Overlap preservation.}
In the state-rotation picture, QRSI replaces $\mathcal{M}$ by
$R_i\mathcal{M}$.
Since $R_i$ is an isometry,
\begin{equation}\label{eq:dist_pres}
  \min_{[\phi]\in R_i\mathcal{M}}
  d_{\mathrm{FS}}\bigl([\phi],\,\mathbb{C}P(\mathcal{G})\bigr)
  =
  \min_{[\phi]\in\mathcal{M}}
  d_{\mathrm{FS}}\bigl([\phi],\,\mathbb{C}P(\mathcal{G}_i)\bigr),
\end{equation}
where $\mathcal{G}_i = R_i^\dagger\mathcal{G}$.
This identifies the correct target manifold for the rotated instance
but does \emph{not} assert a branch-independent distance to
$\mathbb{C}P(\mathcal{G})$.
At the state level, let $[\tilde\phi_i]\in\mathcal{M}$ be the state
found for $H_i$ with foot-point
$[\tilde w_i]\in\mathbb{C}P(\mathcal{G}_i)$.
After basis correction
$[\phi_i]=[R_i\tilde\phi_i]$,
$[w_i]=[R_i\tilde w_i]\in\mathbb{C}P(\mathcal{G})$, so
$d_{\mathrm{FS}}([\phi_i],[w_i])
= d_{\mathrm{FS}}([\tilde\phi_i],[\tilde w_i])$
and $\|\Pi_{\mathcal{G}}\ket{\phi_i}\|^2
= \|\Pi_{\mathcal{G}_i}\ket{\tilde\phi_i}\|^2$:
overlap is preserved within each branch.
The guarantee requires $R_i$ inside the preparation loop;
a naive serial rotation $R_i$ applied \emph{after} preparation
erases overlap by Haar left-invariance
(Sec.~\ref{supp:serial}).

\emph{(ii) Diversity from transitivity.}
$U(N)$ acts transitively on $\mathbb{C}P^{N-1}$, so the corrected
foot-point $[w_i]\in\mathbb{C}P(\mathcal{G})$ changes with each
$R_i$.
For Haar-random rotations, the foot-points are independently and
uniformly distributed on
$\mathbb{C}P(\mathcal{G})\cong\mathbb{C}P^{g-1}$; for $M\ge g$
they span $\mathbb{C}^g$ almost surely
(Proposition~\ref{prop:diversity}).

\begin{figure}[htbp]
  \centering
  \includegraphics[width=0.5\textwidth]{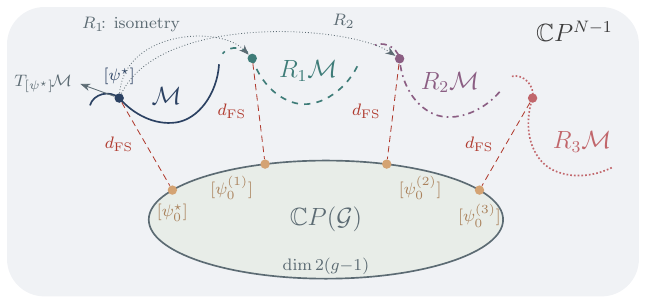}
  \caption{\textbf{Diversity mechanism on $\mathbb{C}P^{N-1}$
    (state-rotation picture).}
    The preparation family $\mathcal{M}$ (solid, dark blue) and its
    isometric copies $R_i\mathcal{M}$ (dashed/dotted) under Haar-random
    rotations~$R_i$.
    Spectral invariance (Proposition~\ref{prop:spec}) guarantees that
    each branch achieves $O(1)$ overlap with $\mathcal{G}$.
    Each copy's optimum projects to a distinct foot-point
    $[\psi_0^{(i)}]\in\mathbb{C}P(\mathcal{G})$ (orange dots),
    and for $M \ge g$ these span~$\mathcal{G}$ almost surely
    (Proposition~\ref{prop:diversity}).
    The Fubini--Study distances $d_{\mathrm{FS}}^{(i)}$ from each
    optimum to its foot-point are generically different across
    branches, reflecting the distinct orientations of $R_i\mathcal{M}$
    relative to $\mathbb{C}P(\mathcal{G})$.
    The tangent-space arrow at $[\psi^\star]$ shows why repeated runs
    of the same preparation family cannot change the foot-point.}
  \label{fig:geometry}
\end{figure}

\section{Why Serial Preparation and Rotation Fails}
\label{supp:serial}

One might attempt to build the ensemble by alternating state preparation
and Haar rotation along a chain: prepare a trial state, rotate the
output, feed the rotated state into a new preparation, rotate again,
and harvest the successive outputs.
However, it is easy to show that this serial scheme cannot jointly
achieve high overlap and high diversity.

Let $G$ be a compact group acting unitarily on $\mathcal{H}$, and
suppose the preparation algorithm produces a state
$\ket{\tilde\psi} = U^*\ket{0}$ with $U^* \in G$.
By left-invariance of the Haar measure, for any fixed $U^* \in G$,
\begin{equation}\label{eq:haar_inv_supp}
  R_i\,U^* \;\sim\; \mathrm{Haar}(G)
  \quad\text{whenever}\quad R_i \sim \mathrm{Haar}(G).
\end{equation}
Consequently, $R_i\ket{\tilde\psi} = R_i U^*\ket{0}$ is distributionally
identical to $R_i\ket{0}$.
In other words, the Haar rotation \emph{erases} any ground-space
overlap that was built by the preparation stage, returning the
effective overlap to the random baseline $O(g/N)$:
\[
  \|\Pi_\mathcal{G}\,R_i\ket{\tilde\psi}\|^2
  \;=\; \|\Pi_\mathcal{G}\,R_i\ket{0}\|^2
  \;\sim\; O(g/N).
\]
Serially alternating preparation and rotation therefore cannot jointly
achieve high overlap and high diversity: state preparation provides
overlap, but the subsequent rotation destroys it; conversely, the
rotation provides diversity, but only at the Haar floor.
This is precisely the diversity / overlap tradeoff, and the serial
scheme does not escape it.

QRSI places the rotation \emph{inside} the
preparation--amplification loop.
Each branch sees a rotated Hamiltonian $H_i = R_i^\dagger H R_i$ and
re-optimizes from scratch; preparation therefore delivers $O(1)$
overlap ab initio on every branch.
In the state-rotation picture, the reachable set $\mathcal{R}$ is
rotated to $R_i\mathcal{R}$ before preparation begins, so the
manifold is translated as a whole and the overlap, as a Fubini-Study
distance, is preserved.
The rotation injects diversity by reorienting the target manifold,
while preparation restores overlap on the reoriented instance: the two
ingredients cooperate rather than compete.

\section{Practical Implementation Considerations}
\label{supp:practical_challenges}

We collect here the per-branch implementation costs for the two
QRSI pictures. QRSI inherits the cost of whichever preparation
primitive is selected and does not amplify it beyond the $M$
parallel invocations counted in the total $M\cdot C_{\mathrm{prep}}$.
The two pictures (Hamiltonian-rotation and state-rotation) impose
qualitatively different constraints on the implementation.

\subsection{Hamiltonian-rotation picture: sparsity destruction}

While the Hamiltonian-rotation picture is mathematically clean, it
faces an implementation obstacle. A Hamiltonian with a sparse or
structured Pauli decomposition $H = \sum_\alpha h_\alpha P_\alpha$
generically loses that structure under conjugation:
$R_i^\dagger P_\alpha R_i$ is a dense unitary, so $H_i$ has $O(4^n)$
nonzero Pauli coefficients even when $H$ had $O(\mathrm{poly}(n))$.
For a Haar-random $R_i$, this sparsity destruction occurs almost
surely, making the rotated Hamiltonian expensive to simulate (whether
by Trotterization, linear combination of unitaries, or quantum signal
processing). Structured alternatives (e.g.\ Givens products,
Householder reflections, random permutations) can preserve more
structure at the cost of weaker randomization; see
Sec.~\ref{supp:rotations}.

\subsection{State-rotation picture: per-primitive cost}

In the state-rotation picture the Hamiltonian is untouched, but the
preparation primitive carries a per-branch cost.

\paragraph{Variational methods (VQE)} inherit any per-call cost of
the chosen ansatz, including the exponential concentration phenomena
(e.g.\ barren plateaus) characteristic of deep or highly expressive
circuits~\cite{mcclean2018barren,larocca2022diagnosing}; QRSI does
not amplify this cost beyond the $M$ branches already counted in the
total $M\cdot C_{\mathrm{prep}}$.

\paragraph{Coherent methods (QPE, QSVT)} provide overlap
amplification at the cost of deep circuits and post-selection; the
per-branch overhead scales as $O(\mathrm{poly}(\log N)/\delta)$.

\paragraph{Projective methods} (imaginary-time
evolution~\cite{Foulkes2001,McArdle2019,Motta2020},
FCIQMC~\cite{booth2009fciqmc,Cleland2010,huggins2022unbiasing,Kanno2023}
or QCQMC~\cite{zhang_2025_qcqmc,Buonaiuto2026}) require convergence
time $O(\gamma^{-1}\log(N/g))$ per branch; moreover, random rotations
can introduce complex Hamiltonian matrix elements that exacerbate
the quantum Monte Carlo sign problem~\cite{booth2009fciqmc}.

\paragraph{Adiabatic methods} must traverse the spectral gap for
each branch, at a cost scaling inversely with the minimum gap along
the adiabatic path. Note that random rotations preserve the
ground-state gap (Proposition 2 in the main text) but not necessarily
the minimum gap along the adiabatic interpolation path.

\medskip

\noindent
QRSI does not avoid these primitive-specific costs, but it ensures
that each invocation operates at $O(1)$ overlap. The total resource
budget is $M \cdot C_{\mathrm{prep}}$ branches (with $M = g$ under
Haar, or $M = O(g\,\eta^{-g}\log(1/\varepsilon))$ under
anti-concentration), plus $O(NM^2)$ for the classical SVD of the
coefficient matrix or $O(M^2/\varepsilon^2)$ randomized-measurement
shots for the Gram matrix on hardware (see Sec.~\ref{supp:cost}).

\section{Structured Rotation Alternatives}
\label{supp:rotations}

When $N$ is large, sampling a full Haar unitary on $U(N)$ may be
prohibitive (cost $O(N^2)$ parameters).
Three structured alternatives provide efficient approximations.

\paragraph{Random Givens rotations.}
A product of $L = O(N \log N)$ random Givens rotations,
\begin{equation}
  R^{(\mathrm{phys})} = \prod_{l=1}^{L} G_{i_l j_l}(\theta_l),
  \quad G_{ij}(\theta) = \mathds{1} - (1-\cos\theta)(e_{ii}+e_{jj})
    + \sin\theta(e_{ij} - e_{ji}),
\end{equation}
forms an approximate 2-design on $U(N)$.
In the state-rotation picture, each Givens factor corresponds to a
single two-qubit gate, so the circuit overhead is $L$ additional
two-qubit gates.
For $n = 10$ qubits ($N \le 1024$), $L \approx 10N$ gates suffices to
approximate the Haar measure to practical precision.

\paragraph{Householder reflections.}
A product of $k$ Householder reflections
$R = \prod_{l=1}^{k}(\mathds{1} - 2\bm{u}_l\bm{u}_l^\dagger)$,
$\bm{u}_l \in \mathbb{C}^N$ unit vectors, provides a tunable
depth--randomness trade-off.
For $k \ge N$, the distribution approaches Haar~\cite{halko2011finding};
for small $k$ the rotation explores a lower-dimensional submanifold but
at reduced cost.

\paragraph{Random permutation matrices.}
Permutation matrices $R = P_\sigma$ for $\sigma \in S_N$ are the
cheapest option (one classical lookup per ensemble member) and require
zero additional circuit depth in the state-rotation picture.
However, they are restricted to a discrete set of $N!$ unitaries
and do not form a continuous design; in practice they can be effective
for small $g$ but may require oversampling ($M \gg g$) to achieve
reliable spanning.

The spanning guarantee of Proposition 1 holds for any distribution over
$U(N)$ that is absolutely continuous with respect to the Haar measure
(which includes the first two alternatives but not permutations in general).

\paragraph{Anti-concentration and practical circuits.}
By Proposition~\ref{prop:diversity_anti_concentration}, any rotation
ensemble satisfying $(\eta,\delta)$-anti-concentration of the
foot-point distribution guarantees diversity with
$M\le g\lceil\eta^{-g}\log(1/\varepsilon)\rceil$ total branches.
This condition is strictly weaker than Haar uniformity and does not
require the rotation ensemble to be a $t$-design. The three
structured alternatives above can be assessed through their
foot-point anti-concentration parameter $\eta$; see
Sec.~\ref{sec:hierarchy} for the detailed analysis.

\section{Concrete Preparation Primitives}
\label{supp:primitives}

The QRSI framework is agnostic to the specific state-preparation
method plugged into each branch.
Figure~\ref{fig:schematic} shows the practical pipeline that wraps any
such primitive, and Table~\ref{tab:primitives} collects the key
properties: the resolvent-power parameter~$q$, the natural rotation
picture, and the per-step cost on both classical and quantum platforms.
Every iterative primitive listed admits a quantum realisation in which
each step is one query to a block encoding of~$H$, implementable via
LCU~\cite{ChildsWiebe2012} or
QSP/QSVT~\cite{LowChuang2017,Gilyen2019} at
$O(\mathrm{poly}(\log N))$ ancilla and depth per
query~\cite{BerryChildsKothari2015}.

\begin{figure}[h]
    \centering
    \includegraphics[width=0.5\textwidth]{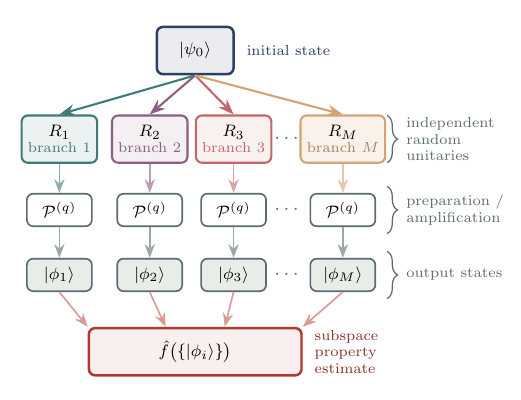}
    \caption{\textbf{Practical schematic of QRSI}. A single initial state
    $\ket{\psi_0}$ is distributed across $M$ independent branches,
    with each `branch' generated by a different instance of a random unitary
    $R_i$. Every branch undergoes the same preparation/amplification
    primitive $\mathcal{P}^{(q)} =: \mathcal{P}^{(q)}_{\omega_i}$
    (e.g.\ imaginary-time evolution, VQE, or adiabatic preparation)
    with $q$ resolvent steps, producing output states $|\phi_i\rangle$
    that concentrate on the target subspace. A post-processing function
    $\hat{f}$ estimates a desired subspace property from the collected states.}
    \label{fig:schematic}
\end{figure}

\begin{table}[h]
  \centering
  \renewcommand{\arraystretch}{1.3}
  \setlength{\tabcolsep}{4pt}
  \begin{tabular}{p{2.8cm}ccclp{3.4cm}}
    \toprule
    \textbf{Primitive}
      & \textbf{Type}
      & \textbf{Picture}
      & $q$
      & \textbf{Classical cost / step}
      & \textbf{Quantum implementation} \\
    \midrule
    Imaginary-time evol.
      & Iterative & Either & \# Trotter steps
      & $O(\gamma^{-1})$
      & QITE~\cite{Motta2020} \\
    VQE / QAOA
      & Iterative & Either & \# opt.\ iterations
      & $c_{\mathrm{eval}}$
      & one $\langle H\rangle$ measurement \\
    Polynomial filter
      & Iterative & State & \# mat-vec products
      & $O(1)$ mat-vec
      & block-encoding query (LCU)~\cite{ChildsWiebe2012} \\
    Chebyshev (KPM)
      & Iterative & State & Polynomial degree $k$
      & $O(1)$ mat-vec
      & QSP / QSVT~\cite{LowChuang2017,Gilyen2019} \\
    Krylov / Lanczos
      & Iterative & State & Krylov dim.\ $m$
      & $O(1)$ mat-vec
      & quantum Krylov~\cite{Stair_2019,Cortes_2022} \\
    Adiabatic
      & Iterative & Hamiltonian & \# time slices
      & $O(\gamma_{\min}^{-1})$
      & Trotterized real-time evol. \\
    QPE
      & Single-shot & Either & $q{=}1$
      & n/a
      & $O(\mathrm{poly}(\log N)/\delta)$~\cite{Kitaev1995,Abrams1999,PoulinWocjan2009} \\
    QSVT filter
      & Single-shot & Either & $q{=}1$
      & n/a
      & $O(\mathrm{poly}(\log N)/\delta)$~\cite{Lin2020,Dong2022} \\
    FCIQMC
      & Stochastic & Either & \# MC sweeps
      & walker propagation
      & QCQMC~\cite{zhang_2025_qcqmc,Buonaiuto2026} \\
    \bottomrule
  \end{tabular}
  \caption{Preparation primitives and their QRSI mapping.
  ``Picture'' indicates the most natural rotation picture
  (Hamiltonian-rotation or state-rotation); ``Either'' means both
  are equally convenient.
  $\gamma$ is the spectral gap, $\gamma_{\min}$ the minimum gap
  along the adiabatic path, and $\delta$ the eigenvalue resolution.
  QPE and QSVT filtering are intrinsically quantum.}
  \label{tab:primitives}
\end{table}

\paragraph{Spectral-filter methods.}
The polynomial (power) filter
$|\psi_i\rangle \propto (\mu I - H)^q R_i|\psi_0\rangle$
is the most literal quantum transcription of classical
RSI~\cite{halko2011finding}, with convergence ratio
$(\gamma/\|H\|)^q$ per step.
The Chebyshev spectral filter~\cite{Jackson_1912,Weisse_2006}
replaces the monomial by $T_k(\tilde H)$, where $\tilde H$ is the
affinely rescaled Hamiltonian, achieving exponential ground-space
amplification once $k\gtrsim\sqrt{\|H\|/\gamma}$.
The Krylov/Lanczos method builds the subspace
$\mathcal{K}_m(H,R_i|\psi_0\rangle)$ and adaptively finds the
optimal degree-$(m{-}1)$ polynomial.
All three are most natural in the \emph{state-rotation} picture,
where $H$ stays fixed and $R_i$ enters only through the probe seed.
On quantum hardware they map to $O(q)$ block-encoding queries via
LCU or QSVT.

\paragraph{Imaginary-time and variational methods.}
Imaginary-time evolution
$|\psi_i(\beta)\rangle\propto e^{-\beta H_i}|\psi_0\rangle$
amplifies the ground-space component by $e^{\beta\gamma}$, with each
Trotter slice being one resolvent power.
Variational algorithms (VQE, QAOA) count optimizer iterations as
resolvent powers.
Both work in either rotation picture; the state-rotation picture is
generally preferred because it preserves the sparsity of~$H$.

\paragraph{Adiabatic preparation.}
Adiabatic state preparation interpolates
$H_i(s) = (1-s)H_0 + s\,R_i^\dagger H R_i$ from a trivial
Hamiltonian~$H_0$ to the target.
This is the only primitive for which the \emph{Hamiltonian-rotation}
picture is unavoidable, since the adiabatic endpoint depends on~$R_i$.

\paragraph{Picture selection rule.}
The state-rotation picture should be the default: it keeps $H$ fixed,
preserving its sparse Pauli decomposition.
The Hamiltonian-rotation picture is required only when the primitive's
target Hamiltonian must change with~$R_i$ (adiabatic preparation).
QRSI's formal guarantees hold identically in both pictures.

\paragraph{Excited-state targeting.}
\label{supp:prim:excited}
QRSI extends to any eigenspace by modifying the spectral filter.
(i)~\emph{Shift-and-invert:} applying
$(H-\sigma I)^{-q}R_i|\psi_0\rangle$ amplifies the eigenvalue
closest to $\sigma$ by $|E_k-\sigma|^{-q}$; sweeping $\sigma$ turns
QRSI into a spectral microscope that reads off each level's
degeneracy from the SVD gap.
(ii)~\emph{Folded spectrum:} replacing $H$ by
$\tilde{H}=(H-\sigma I)^2$ reduces excited-state targeting to a
ground-state problem at the cost of doubling the per-step Hamiltonian
applications.
(iii)~\emph{Chebyshev band-pass filters}~\cite{Weisse_2006} and
\emph{symmetry-sector projections} $H_S = P_S H P_S$ offer additional
routes; the latter is especially useful when conserved quantum numbers
(particle number, total spin) are known.
All four strategies leave the QRSI pipeline structurally intact;
Proposition~\ref{prop:diversity} and all convergence guarantees carry
over unchanged.

\section{Quantitative SVD Gap of the Coefficient Matrix}
\label{supp:svd_gap}

We derive the bound~\eqref{eq:svd_gap_main} stated in the main text.
Write the coefficient matrix as
$C = C_\parallel + C_\perp$, where
$C_\parallel = \Pi_{\mathcal{G}}\,C$ collects the ground-space
components and $C_\perp = \Pi_{\mathcal{G}^\perp}\,C$ collects the
excited-state leakage.  Since $\Pi_{\mathcal{G}}$ has rank~$g$, the
matrix $C_\parallel$ has rank at most~$g$, and therefore
$\sigma_{g+1}(C_\parallel) = 0$.

\paragraph{Upper bound on $\sigma_{g+1}(C)$.}
By Weyl's perturbation inequality for singular
values~\cite{Stewart_1990},
$|\sigma_j(A+B) - \sigma_j(A)| \le \|B\|$ for any matrices $A,B$.
Setting $A = C_\parallel$ and $B = C_\perp$:
\begin{equation}\label{eq:sigma_g1_bound}
  \sigma_{g+1}(C) \;\le\; \sigma_{g+1}(C_\parallel) + \|C_\perp\|
  \;=\; \|C_\perp\|.
\end{equation}
Each column of $C_\perp$ is
$\Pi_{\mathcal{G}^\perp}\ket{\psi_i}$ with
$\|\Pi_{\mathcal{G}^\perp}\ket{\psi_i}\|^2 \le \varepsilon_q$, so
$\|C_\perp\|_F^2 = \sum_{i=1}^{M}
\|\Pi_{\mathcal{G}^\perp}\ket{\psi_i}\|^2 \le M\varepsilon_q$.
Since $\|C_\perp\| \le \|C_\perp\|_F$, we obtain
$\sigma_{g+1}(C) \le \sqrt{M}\,\varepsilon_q^{1/2}$.

\paragraph{Lower bound on $\sigma_g(C)$.}
The foot-point matrix $F \in \mathbb{C}^{g\times M}$, defined by
$F_{k,i} = \braket{v_k}{\psi_i}$, satisfies
$C_\parallel = V_g\,F$ where
$V_g = [\ket{v_1}|\cdots|\ket{v_g}] \in \mathbb{C}^{N\times g}$
has orthonormal columns.
Hence $\sigma_j(C_\parallel) = \sigma_j(F)$ for $j = 1,\dots,g$.
Applying Weyl's inequality in the other direction:
\begin{equation}\label{eq:sigma_g_bound}
  \sigma_g(C) \;\ge\; \sigma_g(C_\parallel) - \|C_\perp\|
  \;=\; \sigma_g(F) - \|C_\perp\|
  \;\ge\; \sigma_g(F) - \sqrt{M}\,\varepsilon_q^{1/2}.
\end{equation}

\paragraph{Positivity of $\sigma_g(F)$.}
Under Haar-random rotations and the nonzero-overlap hypothesis of
Proposition~\ref{prop:diversity},
$\operatorname{rank}(F) = g$ almost surely for $M \ge g$, which
guarantees $\sigma_g(F) > 0$ a.s.
Quantitatively, the columns of $F$ are independent draws from a
continuous distribution on $\mathbb{C}^g$; by the
$U(g)$-equivariance argument (Sec.~\ref{supp:proof:diversity}),
each normalised foot-point is uniform on
$S^{2g-1} \subset \mathcal{G}$.
Writing $F = \mathrm{diag}(\alpha_i)\,\hat{W}$ where $\hat{W}$
has unit-norm columns uniform on $S^{2g-1}$:
\begin{equation}
  \sigma_g(F)
  \;\ge\; \min_i|\alpha_i|\;\sigma_g(\hat{W})
  \;\ge\; (1-\varepsilon_q)^{1/2}\,\sigma_g(\hat{W}).
\end{equation}
The Gram matrix $\hat{W}^\dagger\hat{W}$ concentrates around its
expectation $\frac{M}{g}I_g$ (since
$\mathbb{E}[\hat{w}_i\hat{w}_i^\dagger] = \frac{1}{g}I_g$), so
for $M \ge g + p$ with $p \ge 1$ the minimum singular value
satisfies $\sigma_g(\hat{W}) \ge c\sqrt{p/g}$ with high probability
for a positive constant~$c$ depending on $g$.

\paragraph{SVD gap ratio.}
Combining the bounds:
\begin{equation}\label{eq:svd_gap_ratio}
  \frac{\sigma_g(C)}{\sigma_{g+1}(C)}
  \;\ge\;
  \frac{(1-\varepsilon_q)^{1/2}\,\sigma_g(\hat{W})
    - \sqrt{M}\,\varepsilon_q^{1/2}}
       {\sqrt{M}\,\varepsilon_q^{1/2}}.
\end{equation}
For $\varepsilon_q \ll (1-\varepsilon_q)\,\sigma_g(\hat{W})^2/M$
the first term in the numerator dominates, and the gap grows as
$\varepsilon_q^{-1/2}$.
This mirrors the classical RSI convergence, where the sketch error
decays as $(\sigma_{g+1}/\sigma_g)^{2q}$ with filter power~$q$: in
QRSI the analogous ratio is controlled by the per-branch
leakage~$\varepsilon_q$, which itself decays with the amplification
depth of the chosen primitive.

\section{Resource Analysis}
\label{supp:cost}


Each branch of the QRSI multi-state strategy involves three steps:
\begin{enumerate}
  \item \textbf{Random rotation} ($C_R$): synthesising $R_i$.
    For a product of $L = O(N\log N)$ Givens rotations
    (Sec.~\ref{supp:rotations}), this costs $O(N\log N)$ two-qubit gates.
    In the Hamiltonian-rotation picture, the classical cost of computing
    $H_i = R_i^\dagger H R_i$ is $O(N^2)$.
  \item \textbf{State preparation} ($C_{\mathrm{prep}}$): one invocation of
    the chosen primitive (Sec.~\ref{supp:primitives}).
    Cost depends on the primitive: $O(\gamma^{-1}\log(1/\eta))$ for
    imaginary-time evolution, $O(\mathrm{poly}(\log N)/\delta)$ for QPE,
    or the iteration budget for variational methods.
  \item \textbf{Classical post-processing}: forming and decomposing the
    coefficient matrix $C \in \mathbb{C}^{N\times M}$ costs $O(NM^2)$.
    On quantum hardware, the Gram matrix $G_{ij} = \braket{\psi_i}{\psi_j}$
    can be estimated via SWAP tests or randomized measurements at cost
    $O(M^2/\varepsilon^2)$.
\end{enumerate}


With $M \ge g$ independent branches:
\begin{equation}\label{eq:cost_multi_supp}
  \mathcal{C}_{\mathrm{multi}}
  = M \cdot (C_R + C_{\mathrm{prep}})
  + O(NM^2).
\end{equation}
Since the preparation stage delivers $O(1)$ ground-space overlap per
branch, no further amplification is needed and the cost is dominated by
$M \cdot C_{\mathrm{prep}}$, paid $M = O(g)$ times.

\section{Numerical Verification on a Random Hamiltonian}
\label{supp:random_hamiltonian}

The toric-code example in the main text demonstrates QRSI on a
physically motivated, structured Hamiltonian.  To confirm that the
framework is not specific to topological models, we test it on a
complex-Hermitian random Hamiltonian with \emph{planted} degeneracy
and deliberately structured ground-space eigenvectors.


We construct a $d\times d$ Hermitian matrix $H=V\,\mathrm{diag}(\bm{e})\,V^{\dagger}$ with
$d=256$ and a $g$-fold degenerate ground level at $E_0=0$.
The excited eigenvalues are drawn as $E_0 + \Delta + |z_j|$,
$z_j\sim\mathcal{N}(0,1)$, with gap $\Delta=2$.
The ground-space eigenvectors are \emph{not} random:
\begin{align}
  \ket{v_0} &= \frac{1}{\sqrt{d}}\sum_{k=0}^{d-1}\ket{k}\,,
  \label{eq:v0_rand}\\
  \ket{v_j} &= \frac{\ket{2j{-}1}-\ket{2j}}{\sqrt{2}},
  \quad j=1,\dots,g-1\,.
  \label{eq:vj_rand}
\end{align}
The vector $\ket{v_0}$ is fully delocalized and overlaps every
computational-basis state equally; each $\ket{v_j}$ ($j\ge1$) is
sparse and contributes weight to only two basis vectors.
Consequently, for a generic seed $\ket{e_k}$ (i.e.\ $k\notin\{1,\dots,2(g{-}1)\}$),
none of the sparse eigenvectors overlap it, and the projection
$\Pi_{\mathcal{G}}\ket{e_k}= \braket{v_0}{e_k}\ket{v_0}$
is exactly rank-one.
Since $2(g{-}1) \ll d$, the vast majority of computational-basis seeds
fall in this category, so the ensemble built from independent
random seeds is rank-one with high probability, and projective
filters \emph{without} rotation cannot resolve the full degeneracy.


Figure~\ref{fig:random_hamiltonian} shows the singular-value spectrum
of the QRSI coefficient matrix for $M=20$ branches using
imaginary-time preparation on two instances of this Hamiltonian,
with planted degeneracies $g=6$ and $g=11$.

\begin{figure}[htb!]
  \centering
  \includegraphics[width=\columnwidth]{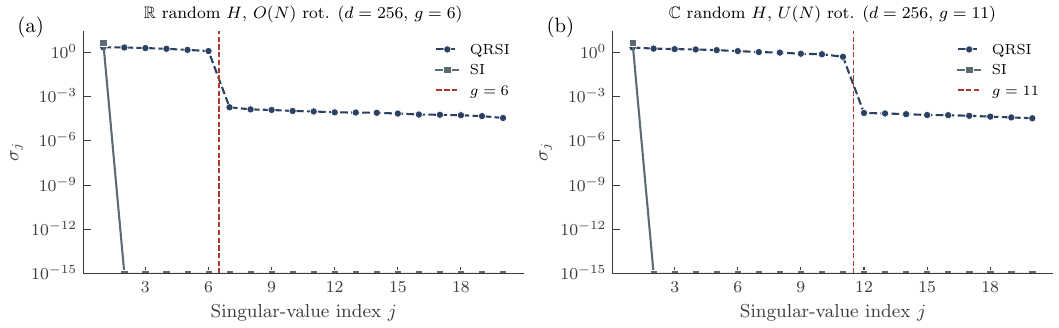}
  \caption{QRSI on structured complex-Hermitian random Hamiltonians
    ($d=256$) with planted ground-state degeneracy.
    \textbf{(a)}~$g=6$: with $U(N)$ Haar rotation (navy circles) a
    clear singular-value gap appears at $j=g$, correctly identifying
    the degeneracy.  Without rotation (slate squares) the ensemble is
    rank-one and no gap is visible.
    \textbf{(b)}~$g=11$: same behaviour at higher degeneracy.
    The dashed orange line marks the planted~$g$.}
  \label{fig:random_hamiltonian}
\end{figure}

With $U(N)$ Haar rotation, a clear spectral gap appears at $j=g$ in
both cases, correctly identifying the planted degeneracy.
Without rotation, all branches collapse to the dominant $\ket{v_0}$
direction (cf.\ Remark~\ref{rem:eigvec_structure}), the ensemble is
rank-one, and no spectral gap emerges.
This confirms that QRSI's diversity guarantee
(Proposition~\ref{prop:diversity}) extends to generic Hamiltonians
beyond topological models, provided the preparation family satisfies
the reachability condition.

\section{Relation to Exact Eigensolvers and Orthogonality-Based Methods}
\label{supp:orthogonality}

A classical exact eigensolver such as \texttt{eigh} resolves a $g$-fold
degenerate eigenspace by returning $g$ mutually orthogonal eigenvectors.
Orthogonality within $\mathcal{G}$ is both necessary and sufficient for
spanning: any $g$ vectors in a $g$-dimensional space that are pairwise
orthogonal form a basis. This observation suggests a natural hierarchy
of diversity mechanisms, ordered from weakest to strongest, which
clarifies the role of randomisation in QRSI.

The weakest requirement is \emph{linear independence}: $g$ vectors in
$\mathcal{G}$ whose projections are linearly independent suffice to
span the eigenspace, but the resulting basis may be arbitrarily
ill-conditioned.
The intermediate requirement is \emph{mutual orthogonality}: this is
exactly what \texttt{eigh}{} and deflation-based quantum methods
(VQD~\cite{higgott2019variational},
SSVQE~\cite{nakanishi2019subspace}) enforce, yielding perfectly
conditioned bases but at the cost of sequential execution.
The strongest requirement is \emph{Haar-uniform directions} on the unit
sphere $S^{2g-1} \subset \mathcal{G}$: this is what QRSI provides via
Proposition~\ref{prop:diversity} of the main text.

The key tradeoff is that enforcing orthogonality requires
\emph{inter-branch communication}: to prepare the $j$-th eigenvector
orthogonal to the preceding $j-1$, one must first know those $j-1$
states.
In VQD, for example, the loss function for the $j$-th state includes
penalty terms $\sum_{k<j} \beta_k |\!\braket{\psi_k|\psi_j}\!|^2$
that require estimating overlaps with all previously converged states,
each via SWAP tests or shadow tomography, an inherently sequential
and measurement-intensive procedure.
SSVQE similarly optimises a weighted sum of energies over an
orthonormal set, coupling all $g$ branches through a single
variational objective.
Exact eigensolvers internalise this coupling through the Gram--Schmidt
or Householder steps of the Lanczos/Arnoldi iteration, which enforce
orthogonality at every step via $O(g)$ inner products.

QRSI replaces this explicit orthogonality enforcement with
randomisation.
The Haar-random rotations $R_i$ produce foot-points
$f(R_i) \in S^{2g-1}$ that are independently and uniformly distributed
over the unit sphere in $\mathcal{G}$ (Proposition~\ref{prop:diversity}).
Since independent uniform vectors in $\mathbb{C}^g$ are almost surely
linearly independent for $M \ge g$, QRSI achieves the minimal
requirement (linear independence, hence spanning)
without any inter-branch coupling.
Moreover, the uniform distribution provides more than bare linear
independence: the expected Gram matrix of $M$ independent uniform
vectors on $S^{2g-1}$ concentrates around
$(1/g)\,I_g + O(1/\sqrt{g})$ off-diagonal entries, so the ensemble is
approximately orthogonal, and the coefficient matrix $C$ is
well-conditioned, without any explicit orthogonalisation step.

Table~\ref{tab:orth_vs_rand} summarises the structural comparison.

\begin{table}[h]
  \centering
  \setlength{\tabcolsep}{5pt}
  \begin{tabular}{lcc}
    \toprule
    & \textbf{Orthogonality enforcement} & \textbf{QRSI (randomisation)} \\
    \midrule
    States needed & Exactly $g$ & $M \ge g$ (mild oversampling) \\
    Inter-branch coupling & Required (deflation / penalty) & None \\
    Parallelisation & Sequential by construction & Embarrassingly parallel \\
    Post-processing & None (already orthonormal) & SVD or rank detection \\
    Overlap estimation & Per pair, per iteration & Only at extraction \\
    Numerical conditioning & Perfect & Near-perfect (concentration) \\
    \bottomrule
  \end{tabular}
  \caption{Structural comparison of orthogonality-based eigenspace
    resolution (as in exact eigensolvers, VQD, SSVQE) and
    randomisation-based resolution (QRSI).}
  \label{tab:orth_vs_rand}
\end{table}

This comparison also illuminates
Proposition~\ref{prop:diversity_anti_concentration} in the main text.
Full Haar uniformity is strictly stronger than what is needed for
spanning: diversity requires only $g$ linearly independent foot-points,
not uniformity on $S^{2g-1}$. The $(\eta,\delta)$-anti-concentration
condition captures this minimal requirement. For the toric-code
example with $g = 4$ and $\eta = 0.97$ (achievable with $\delta = 0.1$
under Haar-like circuits), the sample complexity is
$M \le 24$ branches for $99\%$ spanning probability, showing that the
practical overhead of randomisation over exact orthogonality
enforcement is mild.

The analogy with classical numerical linear algebra is precise.
The Lanczos algorithm~\cite{Liesen_2012} with explicit reorthogonalisation is the
classical counterpart of VQD/SSVQE: it is sequential, requires
$O(g)$ inner products per step, and produces an exactly orthogonal
basis.
Classical randomised subspace iteration~\cite{halko2011finding,
Saad2011eigenvalue} is the counterpart of QRSI: it sketches with a
random probe matrix, amplifies via a spectral filter, and extracts the
basis via an SVD, no explicit reorthogonalisation, trivially
parallelisable, at the cost of mild oversampling and a final rank-revealing
factorisation.
QRSI inherits this same tradeoff from its classical ancestor, elevating
it to the quantum setting where the cost of inter-branch overlap
estimation makes the sequential approach substantially more expensive.

\end{document}